\documentclass[nofootinbib,prb,amsmath,amssymb,notitlepage,floatfix]{revtex4-1}

\usepackage{afterpage}

\usepackage{graphicx,amsfonts}
\usepackage{amsthm}
\usepackage{mathtools} 

\usepackage{esdiff}

\newcounter{subeq}
\renewcommand{\thesubeq}{\theequation\alph{subeq}}
\newcommand{\newsubeqblock}{\setcounter{subeq}{0}\refstepcounter{equation}}
\newcommand{\mysubeq}{\refstepcounter{subeq}\tag{\thesubeq}}

\usepackage{placeins}
\usepackage{bbm}
\usepackage{lmodern,bm}
\usepackage[utf8]{inputenc}
\usepackage{hyperref}
\usepackage{color}
\usepackage{fancyvrb}

\usepackage[normalem]{ulem}

\usepackage{subcaption}
\newcommand{\bor}{\bm{r}}
\newcommand{\born}{\bm{r}_0}
\newcommand{\boa}{\bm{a}}
\newcommand{\bob}{\bm{b}}

\newcommand{\boex}{\bm{x}}
\newcommand{\boy}{\bm{y}}

\newcommand{\bosig}{{\boldsymbol\sigma}}
\newcommand{\opa}{\operatorname{A}}
\newcommand{\opb}{\operatorname{B}}
\newcommand{\opc}{\operatorname{C}}
\newcommand{\opi}{\operatorname{I}}

\newcommand{\opq}{\operatorname{Q}}
\newcommand{\opp}{\operatorname{P}}

\newcommand{\expr}[1]{\left\langle {#1}\right\rangle_\rho} 
\newcommand{\expnr}[1]{\left\langle {#1}\right\rangle} 

\usepackage{xparse}
\DeclareDocumentCommand{\sdev}{ O{} m }{
  {\Delta_{{#1}} {#2}}
}
\DeclareDocumentCommand{\var}{ O{} m }{
  {\Delta^2_{#1} {#2}}
}
\newcommand{\varr}[1]{{\var[\rho]{{#1}}}}
\newcommand{\sdevr}[1]{{\sdev[\rho]{{#1}}}}

\newcommand{\sdevmin}[1]{{\sdev{#1}_\text{min}}}
\newcommand{\varmin}[1]{{\var{#1}_\text{min}}}

\newcommand{\sdevmax}[1]{{\sdev{#1}_\text{max}}}
\newcommand{\varmax}[1]{{\var{#1}_\text{max}}}

\newcommand{\com}[2]{\left[#1,#2\right]} 
\newcommand{\comm}[2]{#1 #2 - #2 #1} 
\newcommand{\acom}[2]{#1 #2 + #2 #1} 

\newcommand{\si}{\mathcal{S}}
\newcommand{\hi}{\mathcal{H}}

\renewcommand{\max}{{\rm max}}
\renewcommand{\min}{{\rm min}}

\renewcommand{\Re}{\operatorname{Re}}
\renewcommand{\Im}{\operatorname{Im}}
\newcommand{\ival}{I}
\usepackage{tikz}
\usepackage{tikz-3dplot}
\usetikzlibrary{positioning,intersections,shadings}

\usepackage{pgfkeys}
\pgfkeys{
  /drawSemiCircle/.is family, /drawSemiCircle,
  defaults/.style = {scale = \textwidth,
    aAngle = 0,
    bAngle = 0,
    rAngle = 0,
    rLength = 1,
    drawBPrime = false,
    drawRightAngles = true,
    rightAngleScale = 0.07},
  scale/.estore in = \scale,
  aAngle/.estore in = \aAngle,
  bAngle/.estore in =\bAngle,
  rAngle/.estore in =\rAngle,
  rLength/.estore in =\rLength,
  drawBPrime/.estore in =\drawBPrime,
  rSub/.estore in =\rSub,
  drawRightAngles/.estore in =\drawRightAngles,
  rightAngleScale/.estore in =\rightAngleScale,
}

\usepackage{xstring}

\newcommand{\DrawSemiCircle[2]}{%
  \pgfkeys{/drawSemiCircle, defaults, #1}%
  \begin{tikzpicture}[scale=\scale/2cm,>=stealth] 
    \draw (-1,0) -- (1,0);

    \clip (-1cm, 0cm) rectangle (1.2cm, 1.2cm);

    \draw (0cm,0cm) circle (1cm);

    \draw [name=vect,thick,->] (0,0) -- node[pos=1,fill=none,label=\aAngle:$\boa$] {} (\aAngle:1cm);
    \draw [name=vect,thick,->] (0,0) -- node[pos=1,fill=none,label=\bAngle:$\bob$] {} (\bAngle:1cm);
    \IfEqCase{\drawBPrime}{
      {true}{
        \pgfmathsetmacro{\bPrimeAngle}{\bAngle-90};
        \draw [name=vect,thick,->] (0,0) -- node[pos=1, left  = 0.1 ,fill=none] {$\bob^\prime$} (\bPrimeAngle:1cm);
        \IfEqCase{\drawRightAngles}{
          {true}{
            \pgfmathsetmacro{\midAngle}{(\bPrimeAngle+\bAngle)/2}
            \pgfmathsetmacro{\midLength}{sqrt(2)*\rightAngleScale}
            \draw [name=rangle] (\bAngle:\rightAngleScale) -- (\midAngle:\midLength) -- (\bPrimeAngle:\rightAngleScale);
          }
          {false}{}
        }
      }
      {false}{}
    }
    \draw [name=vect,thick,->] (0,0) -- node[midway, fill=white] {$\bor_{\rSub}$} (\rAngle:\rLength);

    \pgfmathsetmacro{\ra}{1*\rLength*cos(\aAngle-\rAngle)}
    \pgfmathsetmacro{\rb}{1*\rLength*cos(\bAngle-\rAngle)}

    \draw [name=vect,thick,->] (0,0) -\- node[pos=1, below right = 0.06 and 0.03 ,fill=none] {$\boa^*$} (\aAngle:\ra);
    \draw [name=vect,thick,->] (\aAngle:\ra) -- node[midway,fill=white] {$\bm{x}$} (\rAngle:\rLength);
    \IfEq{\rb}{0}{
      \IfEqCase{\drawRightAngles}{
        {true}{
          \pgfmathsetmacro{\midAngle}{(\rAngle+\bAngle)/2}
          \pgfmathsetmacro{\midLength}{sqrt(2)*\rightAngleScale}
          \draw [name=rangle] (\rAngle:\rightAngleScale) -- (\midAngle:\midLength) -- (\bAngle:\rightAngleScale);
        }
        {false}{}
      }
    }
    {\draw [name=vect,thick,->] (0,0) -- node[pos=1, below left  = 0.05 and 0.03 ,fill=none] {$\bob^*$} (\bAngle:\rb);
      \draw [name=vect,thick,->] (\bAngle:\rb) -- node[midway,fill=white] {$\bm{y}$} (\rAngle:\rLength);}

    \IfEq{\ra}{0}{}{
      \IfEqCase{\drawRightAngles}{
        {true}{
          \pgfmathsetmacro{\rx}{\rLength*cos(\rAngle)}
          \pgfmathsetmacro{\ry}{\rLength*sin(\rAngle)}

          \pgfmathsetmacro{\rax}{\ra*cos(\aAngle)}
          \pgfmathsetmacro{\ray}{\ra*sin(\aAngle)}
          \pgfmathsetmacro{\xxUnNormed}{\rx - \rax}
          \pgfmathsetmacro{\xyUnNormed}{\ry - \ray}
          \pgfmathsetmacro{\xx}{\xxUnNormed / sqrt(\xxUnNormed*\xxUnNormed + \xyUnNormed*\xyUnNormed)}
          \pgfmathsetmacro{\xy}{\xyUnNormed / sqrt(\xxUnNormed*\xxUnNormed + \xyUnNormed*\xyUnNormed)}
          \pgfmathsetmacro{\aStartX}{\rax * (1 + \rightAngleScale/\ra)}
          \pgfmathsetmacro{\aStartY}{\ray * (1 + \rightAngleScale/\ra)}
          \pgfmathsetmacro{\aMidX}{\aStartX  + \rightAngleScale*\xx}
          \pgfmathsetmacro{\aMidY}{\aStartY  + \rightAngleScale*\xy}
          \pgfmathsetmacro{\aEndX}{\rax + \rightAngleScale*\xx}
          \pgfmathsetmacro{\aEndY}{\ray + \rightAngleScale*\xy}
          \draw [name=rangle] (\aStartX,\aStartY) --  (\aMidX,\aMidY) -- (\aEndX,\aEndY);
        }
        {false}{}
      }
    }
    \IfEq{\rb}{0}{}{
      \IfEqCase{\drawRightAngles}{
        {true}{
          \pgfmathsetmacro{\rbx}{\rb*cos(\bAngle)}
          \pgfmathsetmacro{\rby}{\rb*sin(\bAngle)}
          \pgfmathsetmacro{\yxUnNormed}{\rx - \rbx}
          \pgfmathsetmacro{\yyUnNormed}{\ry - \rby}
          \pgfmathsetmacro{\yx}{\yxUnNormed / sqrt(\yxUnNormed*\yxUnNormed + \yyUnNormed*\yyUnNormed)}
          \pgfmathsetmacro{\yy}{\yyUnNormed / sqrt(\yxUnNormed*\yxUnNormed + \yyUnNormed*\yyUnNormed)}
          \pgfmathsetmacro{\bStartX}{\rbx * (1 + \rightAngleScale/\rb)}
          \pgfmathsetmacro{\bStartY}{\rby * (1 + \rightAngleScale/\rb)}
          \pgfmathsetmacro{\bMidX}{\bStartX  + \rightAngleScale*\yx}
          \pgfmathsetmacro{\bMidY}{\bStartY  + \rightAngleScale*\yy}
          \pgfmathsetmacro{\bEndX}{\rbx + \rightAngleScale*\yx}
          \pgfmathsetmacro{\bEndY}{\rby + \rightAngleScale*\yy}
          \draw [name=rangle] (\bStartX,\bStartY) --  (\bMidX,\bMidY) -- (\bEndX,\bEndY);
        }
        {false}{}
      }
    }
  \end{tikzpicture}
}
\usepackage{xparse}
\DeclareDocumentCommand{\DrawSphere}{ m m m }{%
  \begin{tikzpicture}
    \tdplotsetmaincoords{10}{0}
    \tdplotsetrotatedcoords{90}{90}{-90}
    \coordinate (O) at (0,0,0);
    \pgfmathsetmacro{\scale}{0.1cm}

    \pgfmathsetmacro{\aAngle}{{#1}}
    \pgfmathsetmacro{\bAngle}{{#2}}
    \pgfmathsetmacro{\rZeroAngle}{{#3}}
    \pgfmathsetmacro{\rOneAngle}{\aAngle - abs(\aAngle - \rZeroAngle)}
    \pgfmathsetmacro{\rOneLength}{\rOneAngle < 0 ? sqrt( 1 + (cos(\aAngle - \bAngle) / cos(\bAngle - \aAngle + 90))^2  ) *  abs(cos(\rZeroAngle - \aAngle)) : 1}
    \pgfmathsetmacro{\rOneAngle}{\rOneAngle < 0 ? 0 : \rOneAngle}

    \pgfmathsetmacro{\bPrimeAngle}{\bAngle - 90}

    \tdplotsetcoord{a}{\scale}{90}{\aAngle}
    \tdplotsetcoord{b}{\scale}{90}{\bAngle}
    \tdplotsetcoord{bPrime}{\scale}{90}{\bPrimeAngle}
    \tdplotsetcoord{rZero}{\scale}{90}{\rZeroAngle}
    \tdplotsetcoord{rOne}{\rOneLength*\scale}{90}{\rOneAngle}

    \tdplotdrawarc[tdplot_screen_coords]{(O)}{\scale}{0}{360}{}{}

    \draw [thick,->] (O) -- node[pos=1,fill=none, label={\aAngle:$\boa$}] {} (a);
    \IfEq{\bAngle}{\rZeroAngle}{
      \draw [thick,->] (O) -- node[pos=1, fill=none, label={\bAngle:$\bob = \bor_0$}] {} (b);
    }{
      \draw [thick,->] (O) -- node[pos=1, fill=none, label={\bAngle:$\bob$}] {} (b);
      \draw [thick,->] (O) -- node[pos=1, fill=none, label={\rZeroAngle:$\bor_0$}] {} (rZero);
    }

    \IfEq{\bPrimeAngle}{\rOneAngle}{
      \IfEq{\rOneLength}{1}{
        \draw [thick,->] (O) -- node[pos=1, fill=none, label={\bPrimeAngle:$\bob^\prime = \bor_1$}] {} (bPrime);
      }{
        \draw [thick,->] (O) -- node[pos=1, fill=none, label={\bPrimeAngle:$\bob^\prime=\hat\bor_1$}] {} (bPrime);
        \draw [thick,->] (O) -- node[pos=1, fill=none, label={[label distance=0.2*\scale]120:$\bor_1$}] {} (rOne);
      }
    }{
      \draw [thick,->] (O) -- node[pos=1, fill=none, label={\bPrimeAngle:$\bob^\prime$}] {} (bPrime);
      \draw [thick,->] (O) -- node[pos=1, fill=none, label={\rOneAngle:$\bor_1$}] {} (rOne);
    }
    \pgfmathsetmacro{\bRZeroCircleRad}{abs(sin(\bAngle - \rZeroAngle))}
    \pgfmathsetmacro{\bRZeroCircleLoc}{abs(cos(\bAngle - \rZeroAngle))}
    \tdplotdrawarc[tdplot_rotated_coords]{(0,0,\bRZeroCircleLoc*\scale)}{\bRZeroCircleRad*\scale}{0}{180}{anchor=north}{  }
    \tdplotdrawarc[tdplot_rotated_coords, style=dashed]{(0,0,\bRZeroCircleLoc*\scale)}{\bRZeroCircleRad*\scale}{180}{360}{anchor=north}{  }

    \pgfmathsetmacro{\bROneCircleRad}{abs(sin(\bAngle - \rOneAngle))}
    \pgfmathsetmacro{\bROneCircleLoc}{cos(\bAngle - \rOneAngle)}
    \tdplotdrawarc[tdplot_rotated_coords]{(0,0,\bROneCircleLoc*\scale)}{\bROneCircleRad*\scale}{0}{180}{anchor=north}{  }
    \tdplotdrawarc[tdplot_rotated_coords, style=dashed]{(0,0,\bROneCircleLoc*\scale)}{\bROneCircleRad*\scale}{180}{360}{anchor=north}{  }

    \tdplotsetrotatedcoords{\aAngle}{90}{-90}
    \pgfmathsetmacro{\aRCircleRad}{abs(sin(\aAngle - \rZeroAngle))}
    \pgfmathsetmacro{\aRCircleLoc}{abs(cos(\aAngle - \rZeroAngle))}
    \tdplotdrawarc[tdplot_rotated_coords]{(0,0,\aRCircleLoc*\scale)}{\aRCircleRad*\scale}{0}{180}{anchor=north}{  }
    \tdplotdrawarc[tdplot_rotated_coords, style=dashed]{(0,0,\aRCircleLoc*\scale)}{\aRCircleRad*\scale}{180}{360}{anchor=north}{  }
  \end{tikzpicture}
}

\newtheorem{thm}{Theorem}
\newtheorem{defi}{Definition}
\newtheorem{lem}{Lemma}

\newtheorem{cor}{Corollary}

\begin{document}

\title{On Quantum Uncertainty Relations and Uncertainty Regions}

\author{Paul Busch and
  Oliver Reardon-Smith\footnote{\href{mailto:ors510@york.ac.uk}{ors510@york.ac.uk}}}
\affiliation{Department of Mathematics, University of York, York YO10 5DD, UK}
\date{\today}

\begin{abstract}\noindent
  Given two or more non-commuting observables, it is generally not possible to simultaneously assign precise values to each. This quantum mechanical uncertainty principle is widely understood to be encapsulated by some form of uncertainty relation,
  expressing a trade-off between the standard deviations or other measures of uncertainty of two (or more) observables, resulting from their non-commutativity.  Typically, such relations are coarse, and miss important features. It was not until very recently that a broader perspective on quantum uncertainty was envisaged and explored, one  that utilises the notion of an {\em uncertainty region}.  Here we review this new approach, illustrating it with pairs or triples of observables in the case of qubit and qutrit systems. We recall some of the shortcomings of traditional uncertainty relations, and highlight their inability to identify the full uncertainty region. These shortcomings suggest a precautionary note  that, we surmise, ought to accompany the presentation of the uncertainty principle in introductory quantum mechanics courses.
\end{abstract}

\maketitle

\tableofcontents

\section{Introduction}

\noindent
The textbook expression of the uncertainty principle is given by the standard uncertainty relation,
\begin{equation}\label{eq:HUR}
  \sdev{\opa}\,\sdev{\opb}\ge \frac12\bigl| \expnr{\comm{\opa}{\opb}}\bigr|\,.
\end{equation}
Here $\opa,\opb$ are self-adjoint operators representing two observables, whose standard deviations $\sdev{\opa}$, $\sdev{\opb}$ are constrained by the (modulus of the) expectation value of the commutator of $\opa,\opb$. This relation was originally conceived for position and momentum by Heisenberg\cite{heisenberg-ur}, with formal proofs provided by Kennard\cite{kennard-ur} and Weyl\cite{weyl-ur}. The above general form is due to Robertson\cite{robertson-ur}; it was soon strengthened by Schr\"odinger\cite{schrodinger-ur}, who deduced a tighter bound by including the so-called covariance term,
\begin{equation}\label{eq:SUR}
  \var{\opa}\,\var{\opb}\ge\frac14\bigl| \expnr{\comm{\opa}{\opb}} \bigr|^2 + \frac14\bigl(\expnr{\acom\opa\opb}-2\expnr{\opa}\expnr{\opb}\bigr)^2.
\end{equation}
These inequalities were originally presented for vector states from the system's Hilbert space, but also hold for mixed states, represented by density operators $\rho$. We use the standard notation $\expnr{\opa}=\expr{\opa} =\operatorname{tr}[\rho\opa]$ for expectation values and $\var{\opa}=\varr{\opa}=\langle\opa^2\rangle_\rho-\expr{\opa}^2$ for variances. 

For many decades, the task of providing a quantitative statement of the uncertainty principles was considered to be settled by stating the above inequalities.
Still, a closer look shows that these relations do not have all the features one might justifiably require of an uncertainty bound. For instance, in the case of observables with discrete bounded spectra, both \eqref{eq:HUR} and \eqref{eq:SUR} reduce to trivialities: if $\rho$ is an eigenstate of (say) $\opa$, so that $\sdevr{\opa}=0$, the inequalities entail no constraint on the value of $\sdevr{\opb}$. A remedy to this particular deficiency came with the discovery of other forms of uncertainty relations, based on the minimisation of functionals of $\sdev{\opa}, \sdev{\opb}$ other than the product\cite{Pati,physRevA.86.024101}. Another issue lies in the fact that the uncertainty of a quantity may not always be best described by the variance, or, more generally, the moments of its distribution; accordingly, new forms of uncertainty relations have been proven for alternative measures of uncertainty, such as entropies\cite{birbil,deutsch-entropic-ur,maassen-uffink-entropic-ur} or overall width\cite{UH}. We will illustrate another curiosity below: the limiting case of equality in \eqref{eq:HUR} may not always indicate minimum uncertainty.

All these forms of uncertainty relation describe aspects of what we refer to as {\em preparation uncertainty}---they are characterisations of the irreducible uncertainty of the values of observables in any given quantum state. Experimentally, an uncertainty trade-off such as that described by the inequality \eqref{eq:HUR} may be tested by separate runs of {\em accurate} measurements of the two observables $\opa,\opb$ under consideration: the statistics of $\opa$ and $\opb$, measured in distinct ensembles of systems described by the same state, will obey \eqref{eq:HUR} and \eqref{eq:SUR}; hence, if $\opa,\opb$ do not commute, the distributions of $\opa,\opb$ cannot both be arbitrarily sharp.
There is another side to Heisenberg's principle, which concerns the necessary error bounds for any {\em joint approximate measurement} of two observables $\opa,\opb$. The problem of finding rigorous formulations of such {\em measurement uncertainty relations} has become a focus of research efforts in recent years\cite{BLW-coll}. We will not enter this subject here but note that  in a number of case studies, measurement uncertainty relations were found to be deducible from associated preparation uncertainty relations. Hence the ideas and observations made in this paper for the latter may be of use for future investigations of the former.

Rather than asking for bounds on some particular choice of uncertainty functional, such as the product or sum of uncertainties, it is of interest to know the \emph{uncertainty region of $\opa$ and $\opb$}, defined as the whole range of possible value pairs $(\sdevr{\opa},\sdevr{\opb})$. This notion does not seem to have considered until recently when similar problems were envisaged with respect to measurement uncertainty\cite{werner-ang-mom,li-refomulating-uncertainty-principle,zhang-etal-stronger-ur,busch-bullock-qubit}: the concept of {\em error region} was introduced as the set of {\em admissible} pairs of approximation errors for joint measurements of non-commuting quantities\cite{busch-heinosaari-qubit}. Arguably, the content of the uncertainty principle can be captured as a statement concerning the `lower boundary' of the preparation and measurement uncertainty regions:
if $\opa,\opb$ do not commute, these regions cannot, in general, contain all points near the origin of the relevant uncertainty diagrams.

The purpose of the present work is to give an accessible introduction of the subject of uncertainty regions, offering worked examples for pairs of observables in low-dimensional Hilbert spaces. We also explore the logical relation between characterisations of uncertainty regions and standard uncertainty relations.

The  paper is organised as follows. After a brief review of the uncertainty region for the position and momentum of a particle in the line (Section \ref{sec:qp}), we give a general definition of the uncertainty region (Section \ref{sec:pur-def}) and proceed to consider the qubit case in some detail (Section \ref{sec:qubit}). We then proceed to determine uncertainty regions for some pairs of qutrit observables, noting interesting contrasts with the case of qubit observables (Section \ref{sec:qutrit}). We conclude with a summary and some general observations (Section \ref{sec:conclusion}).

\section{Uncertainty regions}

\subsection{Warm-up: a review of position and momentum}\label{sec:qp}

The Heisenberg uncertainty relation for position $\operatorname{Q}$ and momentum $\operatorname{P}$ of a particle on a line is given by the inequality \eqref{eq:HUR},
\begin{equation}\label{eq:qp-ur}
  \sdevr{\opq}\,\sdevr{\opp}\ge\frac{\hbar}{2},
\end{equation}
valid for all states $\rho$ for which both variances are finite. This relation
is tight  in the following sense: for any pair of numbers $(\sdev{\opq}$, $\sdev{\opp})$ with $\sdev{\opq}\,\sdev{\opp}\ge \hbar/2$, there exists a state $\rho$ such that $\sdev{\opq} =\sdevr{\opq}$ and $\sdev{\opp}=\sdevr{\opp}$.  In particular, points of the lower bounding hyperbola branch in the positive quadrant of the $\sdev{\opq}$-$\sdev{\opp}$-plane are realized by pure states, $\rho=|\psi\rangle\langle\psi|$, where the unit vector $\psi$ represents a Gaussian wave function. Moreover, it is not hard to show that every point in the area above the hyperbola can be realized by some quantum state, so that the whole uncertainty region for position and momentum is described by the uncertainty relation \eqref{eq:qp-ur} (Fig. \ref{fig:QP-UR}).

It is interesting to note that the inequality \eqref{eq:qp-ur}  can be equivalently recast  in the form of {\em additive} uncertainty relations: let $\ell>0$ be an arbitrary fixed parameter with the  dimension of length, then
\begin{align}
  \frac{\sdevr{\opq}}{\ell}+\frac{\ell\sdevr{\opp}}{\hbar}&\ge \sqrt2,\label{eq:qp-add-lin}\\
  \frac{\varr{\opq}}{\ell^2}+\frac{\ell^2\varr{\opp}}{\hbar^2}&\ge 1.\label{eq:qp-add-quadr}
\end{align}
The proof of this equivalence follows from an elementary algebraic observation: given arbitrary $\xi,\eta>0$ we have the simple identity
\begin{align}
  \frac \xi x+x \eta=\left(\sqrt{\frac \xi x}-\sqrt{x \eta}\right)^2&+2\sqrt{\xi \eta},
\end{align}
valid for all $x>0$. This quantity assumes its  minimal value $2\sqrt{\xi \eta}$ at $x=\sqrt{\xi/ \eta}$. Therefore, if $C$ is a positive constant, then
\begin{equation}\label{eq:equiv}
  \xi\eta\ge C\quad\iff\quad\forall x>0:\,\frac\xi x+x\eta\ge 2\sqrt{C}.
\end{equation}
Putting
$(\xi,\eta,C)= (\sdev{\opq}/\ell,\ell\sdev{\opp}/\hbar,1/2)$ or $=(\var{\opq}/\ell^2,\ell^2\var{\opp}/\hbar^2,1/4)$ and choosing $x=1$, we see that the uncertainty relation \eqref{eq:qp-ur} entails  \eqref{eq:qp-add-lin} and \eqref{eq:qp-add-quadr}, for every state $\rho$ via the equivalence \eqref{eq:equiv}.

To obtain the reverse implication, we have to make the stronger assumption that one of the additive inequalities, say \eqref{eq:qp-add-lin}, holds {\em for all} $\rho$, for some fixed value $\ell$. To show that then this inequality holds for all $\ell$, we replace $\ell$ with $\ell'\equiv x\ell$,  with $x>0$. Using the unitary scaling transformation,
\begin{equation}
  U_\tau=\exp\left[\frac{i}{2\hbar} \tau(\opq\opp+\opp\opq)\right],\quad \tau=\ln x,
\end{equation}
we have $U_\tau^* \opq U_\tau=\opq/x\equiv \opq_x$, $U_\tau^* \opp U_\tau=x\opp\equiv \opp_x$, and set $\rho_x = U_\tau \rho U_\tau^*$. We then calculate:
\begin{align}
  \frac{\sdevr{\opq}}{x\ell}+\frac{x\ell\sdevr{\opp}}{\hbar}=
  \frac{\sdevr{\opq_x}}{\ell}+\frac{\ell\sdevr{\opp_x}}{\hbar}
  =\frac{\sdev[\rho_x]{\opq} }{\ell}+\frac{\ell\sdev[\rho_x]{\opp}}{\hbar}\ge \sqrt{2}.
\end{align}
Therefore, using \eqref{eq:equiv}, we conclude that \eqref{eq:qp-ur} follows from \eqref{eq:qp-add-lin} (and similarly from \eqref{eq:qp-add-quadr}). This completes the proof.

Geometrically, the limiting case of equality in \eqref{eq:qp-add-lin} represents a family of straight lines tangent to the hyperbola plotted in a $\sdev{\opq}-\sdev{\opp}$-diagram given by $\sdev{\opq}\sdev{\opp}=\hbar/2$; the totality of these tangents defines the hyperbola. Similarly, the second additive inequality bound \eqref{eq:qp-add-quadr} gives a family of ellipses (with axes given by the coordinate axes) tangent to the hyperbola, again defining it (see Fig. \ref{fig:QP-UR}). We conclude that Heisenberg's uncertainty relation or any of its additive equivalents completely determine the position-momentum uncertainty region.

\begin{center}
  \begin{figure*}[ht]
    \includegraphics[width=.6\textwidth]{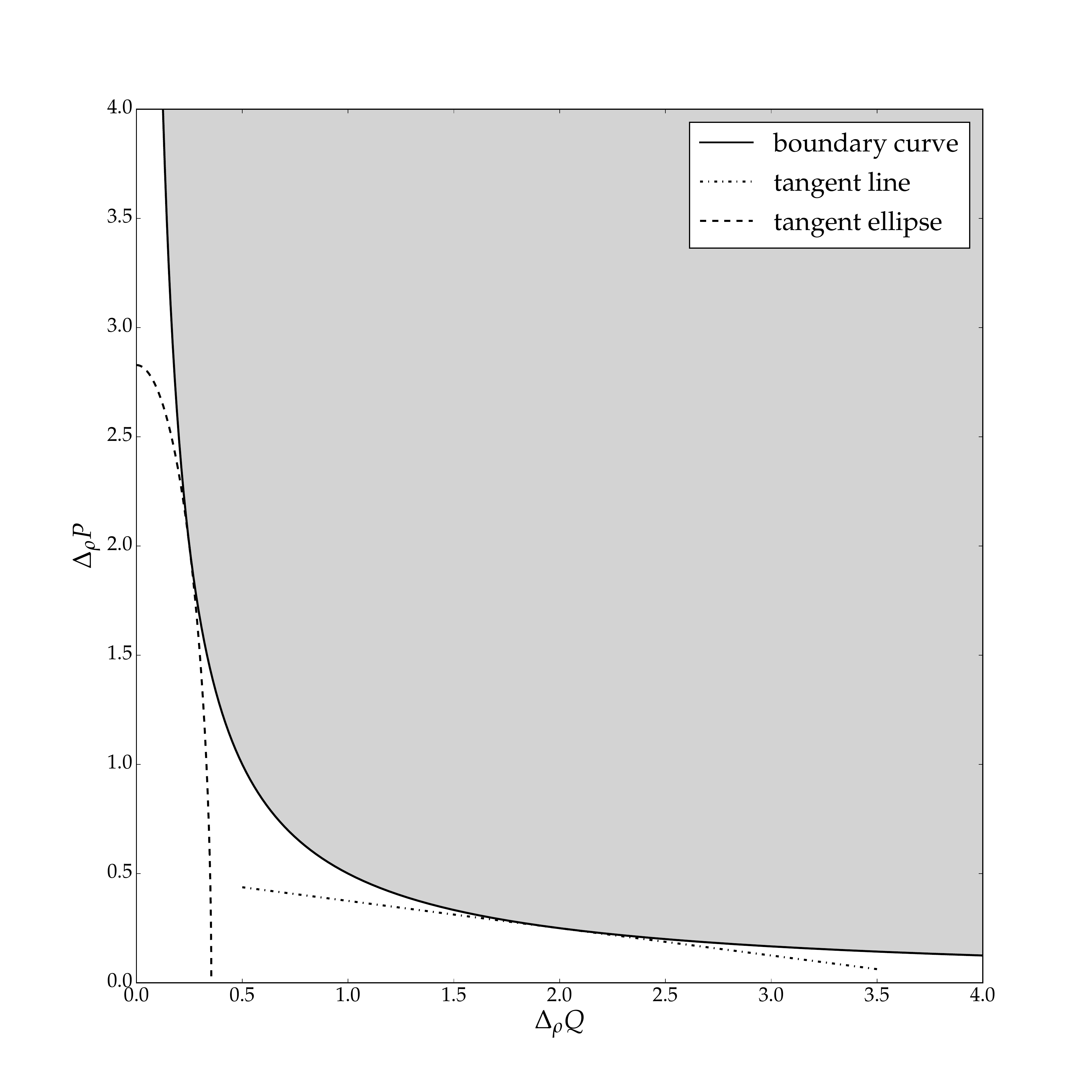}
    \caption{The uncertainty region for the standard deviations of position and momentum (in units where $\hbar = 1$). The solid boundary line represents the hyperbola $\sdev{\opq}\sdev{\opp}=\frac{1}{2}$, and the dash-dotted and dotted curves show examples of the tangential straight and elliptic line segments represented by the bounds given in \eqref{eq:qp-add-lin} and \eqref{eq:qp-add-quadr}, respectively.}
    \label{fig:QP-UR}
  \end{figure*}
\end{center}

\subsection{Uncertainty region: general definition}\label{sec:pur-def}
We seek to explore further the feature of tightness of an uncertainty relation  and so adopt the following definitions (see, e.g., Ref. \onlinecite{abbott-state-indep-qubits}).
We will understand tightness in the sense that the given uncertainty relation fully characterizes the set of admissible uncertainty pairs. In order for an inequality for the uncertainties to achieve this, it is necessary that the only state-dependent terms are the uncertainties themselves; hence such inequalities are of the form $f(\sdev{\opa},\sdev{\opb}, \opa,\opb)\ge 0$. As the reference to $\rho$ can then be dropped, we refer to such uncertainty relations as {\em state-independent}, in line with the terminology introduced in  Ref.~\onlinecite{abbott-state-indep-qubits}.
The term {\em tight} is sometimes used to describe an inequality for a set of variables where the limiting case of equality can be reached for some values; here instead we refer to this situation as {\em saturation} of the inequality. 

\begin{defi}\rm
  The {\em (preparation) uncertainty region} for a pair of observables $\opa$ and $\opb$ is the set of points $(\sdev{\opa}, \sdev{\opb})\in \mathbb{R}^2$  that can be realised by some quantum state, $\rho\in\si(\hi)$, that is,
  \begin{equation}
   {\rm PUR}_\Delta(\opa,\opb)= \bigl\{(\sdev{\opa},\sdev{\opb})\,|\,\,\exists\rho\in\si(\hi)\,:\,\sdev{\opa}=\sdevr{\opa},\ \sdev{\opb}=\sdevr{\opb}\bigr\}.
  \end{equation}
\end{defi}
\begin{defi}\label{def:tight}\rm
  A state-independent uncertainty relation, given by an equality, inequality or set of such, for the uncertainties $\sdevr{\opa}$, $\sdevr{\opb}$ of observables $\opa$, $\opb$ will be called {\em tight} if it is satisfied for exactly the points $(\sdev{\opa}, \sdev{\opb})$ inside the uncertainty region.
\end{defi}
Although we focus here mostly on pairs of observables the definitions may be generalised to $n$ observables in the obvious way. Furthermore, one may also take alternative measures of uncertainty instead of the standard deviations. We will occasionally use the variance instead of the standard deviation, where the former is more appropriate.

It is natural to ask whether the tightness of the inequality \eqref{eq:HUR} (and hence \eqref{eq:SUR}) extends beyond the case of position and momentum. More generally, one can ask for any pair (or family) of observables whether the associated uncertainty region can be characterised by suitable uncertainty relations (which then would be tight).

For the purposes of finding expressions of the uncertainty principle, it is sufficient to focus on specifying the curve defined by fixing the value of $\sdev{\opa}$ and finding $\rho\in\si(\sdev{\opa})\coloneqq\{\rho\,|\,\sdevr{\opa}=\sdev{\opa}\}$ such that $\sdevr{\opb}$ is minimized:
\begin{equation}\label{eq:minuAB}
  \sdevmin{\opb}\equiv\min\left\{\sdevr{\opb}\,|\,\rho \in\si(\sdev{\opa}) \right\}.
\end{equation}
Assuming (as we do henceforth) that the underlying Hilbert space is finite-dimensional, then the set of states $\si(\hi)$ is compact in any norm topology (trace norm, operator norm, etc.). Therefore, the continuity of the map $\rho\mapsto (\sdevr{\opa},\sdevr{\opb})$ ensures that the preparation uncertainty region and the subset of states $\si(\sdev{\opa})$ are  compact. It follows that the minimum in \eqref{eq:minuAB} exists. Hence the uncertainty region has a well-defined  lower boundary curve (and similarly upper and side boundary curves).

We illustrate cases where there are non-trivial \emph{upper} bounds for $\sdevr{\opb}$ for some values of $\sdevr{\opa}$. Additionally, when examining qutrit observables in section \ref{sec:qutrit}, we discover that the uncertainty region is not necessarily of a `simple' shape, such as a convex set. In these cases the uncertainty region is too complicated to be conveniently described by a single inequality, but may be given, for example, in terms of its bounding curves.

An extensive study of uncertainty regions for spin components was undertaken by  Dammeier et al.\cite{werner-ang-mom}
However, the features uncovered in these cases are not representative, as illustrated by the example we examine in section \ref{sec:qutrit}. In particular we note that if the point $(\sdev{\opa},\sdev{\opb} ) = (0,0)$ is in the uncertainty region, then the monotone-closure procedure, taking the set of points $\{(x,y) |\, \exists \rho \text{ s.t. } x\geq \sdevr{\opa},\, y\geq\sdevr{\opb}\}$, employed to great effect in the spin case, has the undesirable property that the closure defined is the \emph{entire} positive quadrant.

A state dependent bound for the joint expectation values of an $n$-tuple of sharp, $\pm 1$-valued observables was given by Kaniewski, Tomamichel and Wehner \cite{PhysRevA.90.012332}. Since a binary probability distribution is entirely characterised by the expectation value this provides an implicit characterisation of the uncertainty region. A complete characterisation of the uncertainty region in terms of variances for pairs of measurements on qubits was given by Li and Qiao\cite{li-refomulating-uncertainty-principle}. However, their relation is an implicit rather than explicit one, with the expectation values and variances of each measurement appearing on both sides of the inequality. Abbott et al\cite{abbott-state-indep-qubits}  then derived the full qubit uncertainty region in a way which more readily generalises to provide (not-necessarily tight) bounds in higher dimensional systems and for more than two observables.

Some analytical, as well as some semianalytical bounds on uncertainty regions were recently given by Szymański and \.{Z}yczkowski\cite{1804.06191}, who also give a method for writing a saturated, state independent bound for a general ``sum of variances'' uncertainty relation as a polynomial root finding problem.

Here we review the case of qubit observables, providing yet another proof of a geometric flavor that immediately focuses on and highlights the extremality property that defines the boundary of the uncertainty region (Section \ref{sec:qubit}). We also investigate to what extent the standard uncertainty relations may or may not characterise the uncertainty region and find that the Schr\"odinger inequality cannot, in general, be cast in a state-independent form as defined here. The examples of pairs of qutrit observables given in Section \ref{sec:qutrit} show that structural features found in the qubit case are no longer present in higher dimensions, for example the uncertainty region for two sharp, $\pm 1$-valued qubit observables contains the origin if and only if they commute.

\section{Optimal qubit uncertainty relations}\label{sec:qubit}

We consider sharp qubit observables with measurement outcomes (eigenvalues) $\pm 1$. These are represented as Hermitian operators (or $2\times 2$-matrices) of the form $\opa=\boa\cdot\bosig=a_x\sigma_x+a_y\sigma_y+a_z\sigma_z$, where vector $\boa$ has Euclidean length  $\|\boa\|=1$ and  $\sigma_x,\sigma_y,\sigma_z$ denote the Pauli matrices on ${\mathbb C}^2$. A general qubit state may be expressed as the density operator
\begin{equation}
  \rho=\frac12(I+\bor\cdot\bosig),\quad \|\bor\|=r\le 1,
\end{equation}
where $\opi$ denotes the identity operator (unit matrix). Note that $\rho$ is a pure state if and only if $\|\bor\|=1$.

For $\opa=\boa\cdot\bosig$, we have  $\expr{\opa}=\boa\cdot\bor$ and, since $\opa^2=\opi$,
\begin{equation}
  \varr{\opa}=1-(\boa\cdot\bor)^2=1-\|\bor\|^2+\|\bor\times\boa\|^2.
\end{equation}

We recall  that for unit 3-vectors $\boa$ and $\bob$ separated by angle $\theta$ we have $\boa\cdot\bob=\cos\theta$ and $\|\boa\times\bob\| = \sin\theta$.
We also note  the operator norm of the commutator of $\opa,\opb$ is given by
\begin{equation}
  \bigl\|[\opa,\opb]\bigr\| = 2\|\boa\times\bob\|,
\end{equation}
which suggests that $\sin\theta=\|\boa\times\bob\|$ is the relevant quantity to measure the degree of noncommutativity (incompatibility) of $\opa$ and $\opb$.

\subsection{Uncertainty bounds for $\sigma_x,\sigma_y,\sigma_z$}\label{sec:xyz}
Considering the variances of $\sigma_x,\sigma_y,\sigma_z$ in a general state $\rho$,
\begin{equation}
  \varr{\sigma_x}=1-r_x^2,\quad \varr{\sigma_y}=1-r_y^2,  \quad \varr{\sigma_z}=1-r_z^2,
\end{equation}
it is easy to see that the positivity condition for $\rho$, $r_x^2+r_y^2+r_z^2\le 1$, is in fact equivalent to the following additive uncertainty relation:
\begin{align}\label{eq:urxyz}
  \varr{\sigma_x}+\varr{\sigma_y}+\varr{\sigma_z}=3-\|\bor\|^2\ge 2.
\end{align}
This inequality is saturated if and only if $\rho$ is a pure state ($r^2=1$).
Given that the standard deviations $\sdevr{\sigma_k}\in[0,1]$, the uncertainty region for the triple $(\sigma_x,\sigma_y,\sigma_z)$ is given as the complement of the open ball at the origin with radius $\sqrt{2}$ intersected with the unit cube $[0,1]\times[0,1]\times[0,1]$ (Fig.~\ref{fig:orth-triple}). The inequality \eqref{eq:urxyz} is an instance of a general triple uncertainty relation for the components of a spin-$s$ system, with the bound for the sum of variances being $s$, as shown in Ref.~\onlinecite{Hofmann03}.

\begin{center}
  \begin{figure*}[ht]
    \includegraphics[width=.5\textwidth]{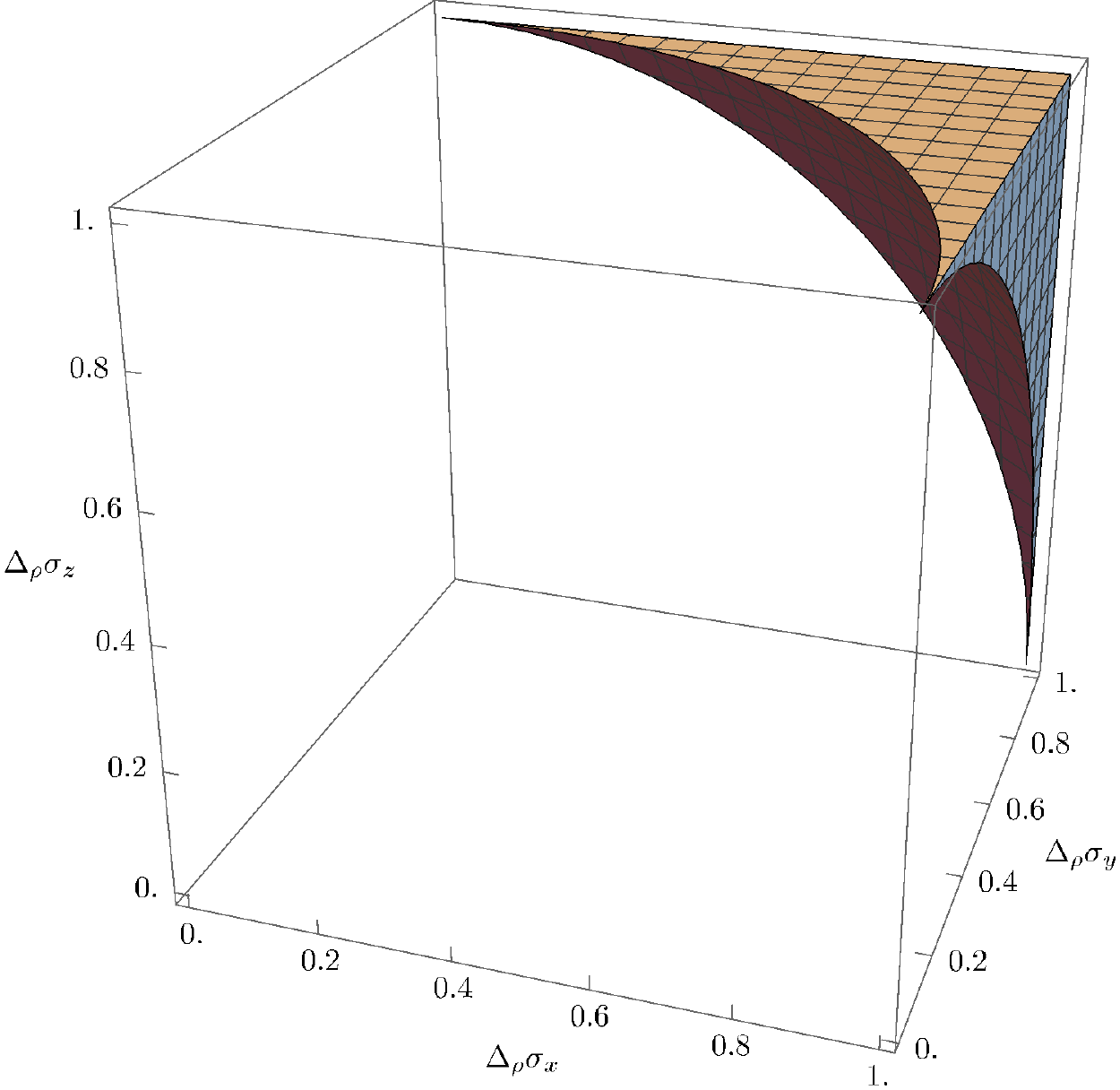}
    \caption{The uncertainty region for the standard deviations of the qubit triple $(\sigma_x,\sigma_y,\sigma_z)$. Note that the top surface with $\sdev{\sigma_z}=1$, shows the uncertainty region for the pair $\sigma_x,\sigma_y$, determined by
    $\varr{\sigma_x}+\varr{\sigma_y}\ge 1$, which can be filled with states whose Bloch vectors have component $r_z=0$.}
    \label{fig:orth-triple}
  \end{figure*}
\end{center}

We briefly revisit and compare the Heisenberg and Schr\"odinger inequalities for spin components.
The uncertainty relation \eqref{eq:HUR} for $\sigma_x,\sigma_y$ is equivalent to
\begin{equation}\label{eq:HUR-sig}
  \varr{\sigma_x}\,\varr{\sigma_y}\ge |\langle \sigma_z\rangle_\rho|^2=1-\varr{\sigma_z}.
\end{equation}
We note that the lower bound on the right hand side becomes zero if $r_z=0$, in which case this inequality gives no constraint on the variances on the left hand side. However, using  \eqref{eq:urxyz}, we obtain
\begin{equation}\label{eq:prodbd}
  \varr{\sigma_x}\,\varr{\sigma_y}\ge\varr{\sigma_x}\bigl(2-\varr{\sigma_x}-\varr{\sigma_z}\bigr)\ge\varr{\sigma_x}\bigl(1-\varr{\sigma_x}\bigr).
\end{equation}
We see that the bound is now nontrivial for all $\rho$ with $\sdevr{\sigma_x}=\sdev{\sigma_x}\in(0,1)$. It can be as large as $1/4$, which is obtained when $\bor=(\pm1,\pm1,0)/\sqrt2$. The above inequality is equivalent to the following, which is also entailed directly by \eqref{eq:urxyz} bearing in mind that $\varr{\sigma_z}\le 1$:
\begin{equation}\label{eq:urxy}
  \varr{\sigma_x}+\varr{\sigma_y}\ge 1.
\end{equation}
From this, we can straightforwardly obtain the minimum \eqref{eq:minuAB} for $\varr{\sigma_y}$ given a fixed $\varr{\sigma_x}\in(0,1)$.
In fact, $\varr{\sigma_y}$ is minimized  when $\varr{\sigma_y}=1-\varr{\sigma_x}$. This is equivalent to $r_x^2+r_y^2=1$, which entails $r_z=0$, that is, $\varr{\sigma_z}=1$. This means, in particular, that  the bound given by \eqref{eq:HUR-sig} becomes trivial and that given by \eqref{eq:SUR} is tight.

It is instructive to consider the conditions under which the Heisenberg inequality \eqref{eq:HUR-sig} is saturated. This gives
$(1-r_x^2)(1-r_y^2)=r_z^2$, or $1+r_x^2r_y^2=r_x^2+r_y^2+r_z^2$. Since the right-hand side is never greater than $1$ and the left hand side never less, both sides must be equal to $1$ and, therefore either $r_x=0$ or $r_y=0$. If $r_x$ is fixed and  non-zero, then $r_y=0$, which is to say that $\varr{\sigma_y}=1$.

Note that here $\varr{\sigma_y}$  is maximal rather than minimal. Saturation of the standard uncertainty relation for these observables thus leads to maximising the uncertainty product instead of minimising it, as one might, naively, have expected. In contrast, equality in \eqref{eq:urxy} forces minimality of the uncertainty product. We also see that \eqref{eq:HUR-sig} forces maximality $\varr{\sigma_x}=\varr{\sigma_y}=1$ by requiring minimal uncertainty for $\sigma_z$, whereas \eqref{eq:urxy} does not stipulate this.

Taking into account the natural upper bound of 1 for the variances, the  inequality \eqref{eq:urxy} is tight, that is, it captures exactly the uncertainty region for $\sigma_x,\sigma_y$, while \eqref{eq:HUR-sig} does not. (As we observed above, \eqref{eq:HUR-sig} does not set a positive lower bound for $\varr{\sigma_y}$ when $\varr{\sigma_z}=1$.)

Since $0\le\bigl(1-\varr{\sigma_x}\bigr)\bigl(1-\varr{\sigma_y}\bigr)$, we have
\begin{equation}
  \varr{\sigma_x}\varr{\sigma_y}\ge \varr{\sigma_x}+\varr{\sigma_y}-1\ge1-\varr{\sigma_z},
\end{equation}
where the latter inequality is obtained from $\|\bor \|^2\le 1$ or the equivalent relation \eqref{eq:urxyz}. It follows that, just like \eqref{eq:urxy}, \eqref{eq:HUR-sig} is also a consequence of (and indeed weaker than) \eqref{eq:urxyz}.

The fact that \eqref{eq:urxyz} implies \eqref{eq:HUR-sig} should not be surprising once one realises that the former inequality  is indeed equivalent to the Schr\"odinger relation \eqref{eq:SUR}, which takes the following form in the present case
\begin{equation}
  \varr{\sigma_x}\,\varr{\sigma_y}\geq \langle\sigma_z\rangle_\rho^2+\langle\sigma_x\rangle_\rho^2\langle\sigma_y\rangle_\rho^2=\bigl(1-\varr{\sigma_z}\bigr)+
  \bigl(1-\varr{\sigma_x}\bigr)\bigl(1-\varr{\sigma_y}\bigr).
\end{equation}
This is equivalent to
\begin{equation}
  \varr{\sigma_x}\,\varr{\sigma_y}\geq 2-\bigl(\varr{\sigma_x}+\varr{\sigma_y}+\varr{\sigma_z} \bigr) +\varr{\sigma_x}\,\varr{\sigma_y},
\end{equation}
and hence to \eqref{eq:urxyz}, and ultimately to $\|\bor\|^2\le 1$, anticipating the results of section \ref{sec:sch-uncertainty-relation}.

We summarise:
\begin{enumerate}
  \setlength\itemsep{0pt}
\item[(1)]
  The Schr\"odinger inequality (but not the Heisenberg inequality) for $\sigma_x,\sigma_y$ determines their uncertainty region.
  
  \item[(2)]
  Saturation of the Heisenberg inequality does not entail minimal, but instead maximal, uncertainty  (i.e., maximal $\sdev{\sigma_y}$ given $\sdev{\sigma_x}\not\in\{0,1\}$).

\item[(3)]
  The uncertainty region for $\sigma_x,\sigma_y$ is the intersection of the unit square $[0,1]\times[0,1]$ with the complement of the open unit ball $\var{\sigma_x}+\var{\sigma_y}< 1$. The lower boundary is reached exactly for pure states with $\sdev{\sigma_z}=1$, which entails the vanishing of the commutator term in the Heisenberg inequality \eqref{eq:urxy} (which therefore becomes trivial at minimum uncertainty). In this case, one has equality in the Schr\"odinger relation, and the uncertainty bound is found to be entirely due to the covariance term.
  
  \item[(4)] The Schr\"odinger inequality, due to its equivalence with \eqref{eq:urxyz}, also determines the triple uncertainty region for $\sigma_x,\sigma_y,\sigma_z$, Fig.~\ref{fig:orth-triple}. 
  
  \item[(5)]
  Saturation of the Schr\"odinger inequality and  the equivalent triple uncertainty relation \eqref{eq:urxyz} is given exactly on the set of pure states. Hence, all pure states are minimum uncertainty states for the triple $\sigma_x,\sigma_y,\sigma_z$.

\end{enumerate}

\subsection{Uncertainty region for general $\pm 1$-valued qubit measurements}
\label{sec:qubit-uncertainty}
We  now consider general observables represented as $\opa=\boa\cdot\bosig$, $\opb=\bob\cdot\bosig$ where $\boa$ and $\bob$ are unit vectors but no longer assumed to be orthogonal. Observables of this form are sufficient to explore the shapes of uncertainty regions since any two outcome qubit observable may be simulated by one of this form using classical post processing (relabelling the measurement outcomes $\pm1$ and adding classical noise).

We begin by noting some simple known examples of state-independent uncertainty relations for the pair $\opa,\opb$ given in Ref.~\onlinecite{BLW14}:
\begin{align}
  \label{eqn:straight-lower-bound-blw}
  \sdev{\opa}+\sdev{\opb}\ &\ge \ \tfrac12\bigl\|[\opa,\opb]\bigr\|,\\
  \label{eqn:curved-lower-bound-blw}
  \var{\opa}+\var{\opb}\ &\ge \ 1-\sqrt{1-\tfrac14\bigl\|[\opa,\opb]\bigr\|^2}.
\end{align}
While these are easily proven by elementary means, it is equally easy to see that they are not tight; they only touch the actual lower boundary curve of the uncertainty region at isolated points. Nevertheless they are of a simple form and illustrate the concept of a state-independent uncertainty bound.

In the following considerations we will make use of the identity
\begin{equation}\label{eq:r2-id}
  \left\|\boa\times\bob\right\|^2 \left\|\bor\right\|^2 = \bigl((\boa\times\bob)\cdot\bor\bigr)^2 + \bigl\|(\boa\times\bob)\times\bor\bigr\|^2 ,
\end{equation}
which is a version of Pythagoras' law. In the special case that $\|\boa\|=\|\bob\|=\|\bor\|=1$ and $\bor\perp\boa\times\bob$, this yields
\begin{equation}
  \left\|\boa\times\bob\right\| = \bigl\|\bob(\boa\cdot\bor)-\boa(\bob\cdot\bor)\bigr\| = \bigl\|\boa(\boa\cdot\bor)-\bob(\bob\cdot\bor)\bigr\|.
\end{equation}
We obtain the following Lemma.
\begin{lem}\rm
  Let $\boa,\bob$ be unit vectors spanning a plane $P$. For any unit vector $\bor\in P$, denote $\boa^*=\boa(\boa\cdot\bor)$,  $\bob^*=\bob(\bob\cdot\bor)$ and $\boex=\bor-\boa^*$, $\boy=\bor-\bob^*$. Then
  \begin{equation}\label{eq:axb}
    \bigl\|\boa\times\bob\bigr\| = \|\boa^*-\bob^*\|=\|\boex-\boy\|.
  \end{equation}
\end{lem}

\begin{figure}[h]
  \begin{subfigure}[t]{0.45\textwidth}
    \DrawSphere{60}{90}{90}
    \caption{$\boa\cdot\bor > \boa\cdot\bob \iff \sdev{\opa} < \|\boa\times\bob\|$}
  \end{subfigure}
  \begin{subfigure}[t]{0.45\textwidth}
    \DrawSphere{60}{90}{90}
    \caption{$\boa\cdot\bor = \boa\cdot\bob > \boa\cdot\bob^\prime = \|\boa\times\bob\| \iff \sdev{\opa} = \|\boa\times\bob\|$}
  \end{subfigure}
  \begin{subfigure}[t]{0.45\textwidth}
    \DrawSphere{60}{90}{120}
    \caption{$\boa\cdot\bor = \boa\cdot\bob^\prime = \|\boa\times\bob\| \iff \sdev{\opa} =\boa\cdot\bob$}
  \end{subfigure}
  \begin{subfigure}[t]{0.45\textwidth}
    \DrawSphere{60}{90}{130}
    \caption{$0 < \boa\cdot\bor < \boa\cdot\bob^\prime \iff \sdev{\opa} > \boa\cdot\bob$}
  \end{subfigure}
  \label{fig:semicircle-min}
  \caption{Determining the locations of Bloch vectors $\bor_0,\bor_1$ (in the plane spanned by $\boa,\bob$) for states minimizing and maximizing $\varr{\opb}$ within the set of states with $\sdevr{\opa}=\sdev{\opa}$.
  We consider the case $\boa\cdot\bob > \|\boa\times\bob\|$ only (shown here for $\theta = \frac{\pi}{6}$). }\label{fig:min-max}
\end{figure}

To describe the uncertainty region, we set out to determine the maximum and minimum values of $\varr{\opb}$ given a fixed value $\var{\opa}$ of $\varr{\opa}$.
Minimality (maximality) of $\varr{\opb}$   is equivalent to maximality (minimality) of $(\bor\cdot\bob)^2$ whilst keeping $(\bor\cdot\boa)^2=1-\var{\opa}$ constant. Hence we are looking for the optimisers within the disks that are the intersections of the planes $\bor\cdot\boa=\pm\sqrt{1-\var{\opa}}$ with the Bloch sphere. We may assume $\boa\cdot\bob\ge0$. For the determination of minimal and maximal $\sdevr{\opb}$  it will be sufficient to focus on the disks of constant $\sdev{\opa}$ with $\bor\cdot\boa\ge0$, and look for the two disks of constant  $\bor\cdot\bob$ which intersect the former disk in just one point. The resulting vectors $\bor_0,\bor_1$ (which are or can be chosen to lie in the plane spanned by $\boa,\bob$) are those giving the largest, resp. smallest, non-negative value of $\bob\cdot\bor$ within the disk of vectors satisfying $\bor\cdot\boa=\sqrt{1-\var{\opa}}$.

\begin{figure}[ht]
  \begin{subfigure}[t]{0.45\textwidth}
    \DrawSemiCircle[aAngle=40,bAngle=107.5,rAngle=70, scale=\textwidth,rSub=0]{}
    \caption{Cross-section through the Bloch sphere showing the relations between the various vectors in the case where the choice of $\sdevr{\opa}=\|\boex\|$ fixes $\bor_0$ to be between $\boa$ and $\bob$ to minimize $\sdevr{\opb}=\|\boy\|$. Note that $\boex\cdot\boy=-\boa\cdot\bob\,\|\boex\|\|\boy\|$.}
    \label{fig:semicircle-min-r-between}
  \end{subfigure}
  \begin{subfigure}[t]{0.45\textwidth}
    \DrawSemiCircle[aAngle=40, bAngle=70, rAngle=110, scale=\textwidth,rSub=0]{}
    \caption{Cross-section through the Bloch sphere showing the relations between the various vectors in the case where the choice of $\sdevr{\opa}$ fixes $\bor_0$ to be outside $\boa$ and $\bob$ to minimize $\sdevr{\opb}$. Note that $\boex\cdot\boy=\boa\cdot\bob\,\|\boex\|\|\boy\|$.}
    \label{fig:semicircle-min-r-outside}
  \end{subfigure}
  \caption{Illustration of the $\born$ vectors which minimize $\sdevr{\opb}$ given a fixed value of $\sdevr{\opa}$}
  \label{fig:semicircle-min-plane}
\end{figure}
\begin{figure}[ht]
  \begin{subfigure}[t]{0.45\textwidth}
    \DrawSemiCircle[aAngle=80,bAngle=110,rAngle=40, scale=\textwidth,drawBPrime=true,rSub=1]{}
    \caption{Upper bound diagram. Note that $\boex\cdot\boy=\boa\cdot\bob\,\|\boex\|\|\boy\|$.}
  \end{subfigure}
  \begin{subfigure}[t]{0.45\textwidth}
    \DrawSemiCircle[aAngle=80, bAngle=110, rAngle=20, scale=\textwidth,rLength=0.85,rSub=1]{}
    \caption{Upper bound diagram version where $\sdev{\opa} \ge \boa\cdot\bob$. Here we can achieve the (trivial) upper bound of $\sdev{\opb} = 1$ since $\bor\cdot\bob=0$. Note that if one wants to have $\rho$ a pure state one can obtain this by moving perpendicularly out of the $\boa$, $\bob$ plane.}
  \end{subfigure}
  \caption{Illustration of the $\bor_1$ vectors which maximize $\sdevr{\opb}=\|\boy\|$ given a fixed  $\sdevr{\opa}=\|\boex\|$}
  \label{fig:semicircle-max-plane}
\end{figure}

Figure \ref{fig:min-max}  shows the Bloch vectors $\bor_0,\bor_1$ corresponding to the states that minimize, resp. maximize $\sdevr{\opb}$. These are unit vectors in the plane spanned by $\boa,\bob$ (except for the cases where the maximum $\sdevr{\opb}=1$ and $\sdev{\opa} > \boa\cdot\bob$). There are four constellations of interest according to distinct regions of increasing values of $\sdev{\opa}$. We determine the minimal and maximal values $\sdevmin{\opb},\sdevmax{\opb}$ in each case. To be specific, we assume $\boa\cdot\bob\ge\|\boa\times\bob\|$, that is, $\theta\le\pi/4$; the case $\theta>\pi/4$ is treated similarly. As evident from Figures \ref{fig:semicircle-min-plane} and \ref{fig:semicircle-max-plane}, we have $\|\boex\|=\sdev{\opa}$, and $\|\boy\|=\sdevmin{\opb}$, resp. $\|\boy\|=\sdevmax{\opb}$. Then repeated application of Eq.~\ref{eq:axb} and using the relation $\boex\cdot\boy=\pm\boa\cdot\bob\|\boex\|\|\boy\|$ with the appropriate choice of sign gives the following equations for $\sdevmin{\opb}$, $\sdevmax{\opb}$
\begin{align}
  \newsubeqblock
  \label{eqn:qubit-bounds-1}
  \mysubeq 0\leq\sdev{\opa} \leq\|\boa\times\bob\|\iff &1\geq \boa\cdot\bor\geq\boa\cdot\bob\\
  \mysubeq  \implies&\|\boa\times\bob\|^2=\var{\opa}+\varmin{\opb}+2\boa\cdot\bob\,\sdev{\opa}\sdevmin{\opb}\\
  \newsubeqblock
  \label{eqn:qubit-bounds-2}
  \mysubeq \|\boa\times\bob\|\leq\sdev{\opa} \leq 1\iff &\boa\cdot\bob\geq\boa\cdot\bor\geq 0\\
  \mysubeq  \implies&\|\boa\times\bob\|^2=\var{\opa}+\varmin{\opb}-2\boa\cdot\bob\,\sdev{\opa}\sdevmin{\opb}\\
  \newsubeqblock
  \label{eqn:qubit-bounds-3}
  \mysubeq 0\leq\sdev{\opa} \leq\boa\cdot\bob\iff &1\geq\boa\cdot\bor\geq\|\boa\times\bob\|\\
  \mysubeq  \implies&\|\boa\times\bob\|^2=\var{\opa}+\varmax{\opb}-2\boa\cdot\bob\,\sdev{\opa}\sdevmax{\opb}\\
  \newsubeqblock
  \label{eqn:qubit-bounds-4}
  \mysubeq \sdev{\opa}\geq\boa\cdot\bob\iff&\boa\cdot\bor\leq\|\boa\times\bob\|\\
  \mysubeq \implies&\sdevmax{\opb}=1\quad\text{at }\bor=\bor_1\ (\text{in the direction of }\bob').
\end{align}
The presence of an upper bound less than 1 for $\sdev{\opb}$ in the cases shown in \eqref{eqn:qubit-bounds-1} and \eqref{eqn:qubit-bounds-2}(i.e., $\sdev{\opa}< \boa\cdot\bob$) can be interpreted in terms of another observable $\opb^\prime = \bob^\prime\cdot\bosig$ where $\bob^\prime$ is the unit vector, in the plane spanned by $\boa$ and $\bob$, orthogonal to $\bob$. With this definition we have $\varr\opb = 1-\varr{\opb^\prime}$ so the lower bound on the uncertainty  $\sdev{\opb^\prime}$ (due to its trade-off with $\sdev{\opa}$) imposes an upper bound on $\sdev{\opb}$.

By solving the various quadratic equations  we obtain the following relation which defines, exactly, the allowed uncertainty region. In particular we can achieve our aim of giving a closed form for the minimum and maximum values of $\sdevr{\opb}$  given a fixed $\sdev{\opa}$. Note that the resulting tight uncertainty relation is state-independent in the sense described above: the bounds for $\varr{\opb}$ depend only on the chosen $\var{\opa}$ and the observables $\opa$ and $\opb$.
\begin{thm}\label{thm:ab-pur}\rm
  Given a pair of qubit observables  $\opa=\boa\cdot\bosig$, $\opb=\bob\cdot\bosig$ as well as a fixed uncertainty $\sdevr{\opa}=\sdev{\opa}$ we have
  \begin{align}
    \left|\sdev{\opa} (\boa\cdot\bob) - \|\boa\times\bob\|\sqrt{1-\var \opa}\right|
    \leq \sdevr{\opb} 
                            \leq\begin{cases}
                              \sdev{\opa} (\boa\cdot\bob) + \|\boa\times\bob\|\sqrt{1-\var \opa} & \text{\normalfont if } \sdev{\opa} < \boa\cdot\bob\\
                              1 & \text{\normalfont otherwise}
                            \end{cases}.
  \end{align}
\end{thm}
The resulting uncertainty region is plotted in Fig.~\ref{fig:top-level-err-region}.

\afterpage{\clearpage}

\begin{figure*}[ht]
  \begin{subfigure}[b]{0.4\textwidth}
    \includegraphics[width=\textwidth]{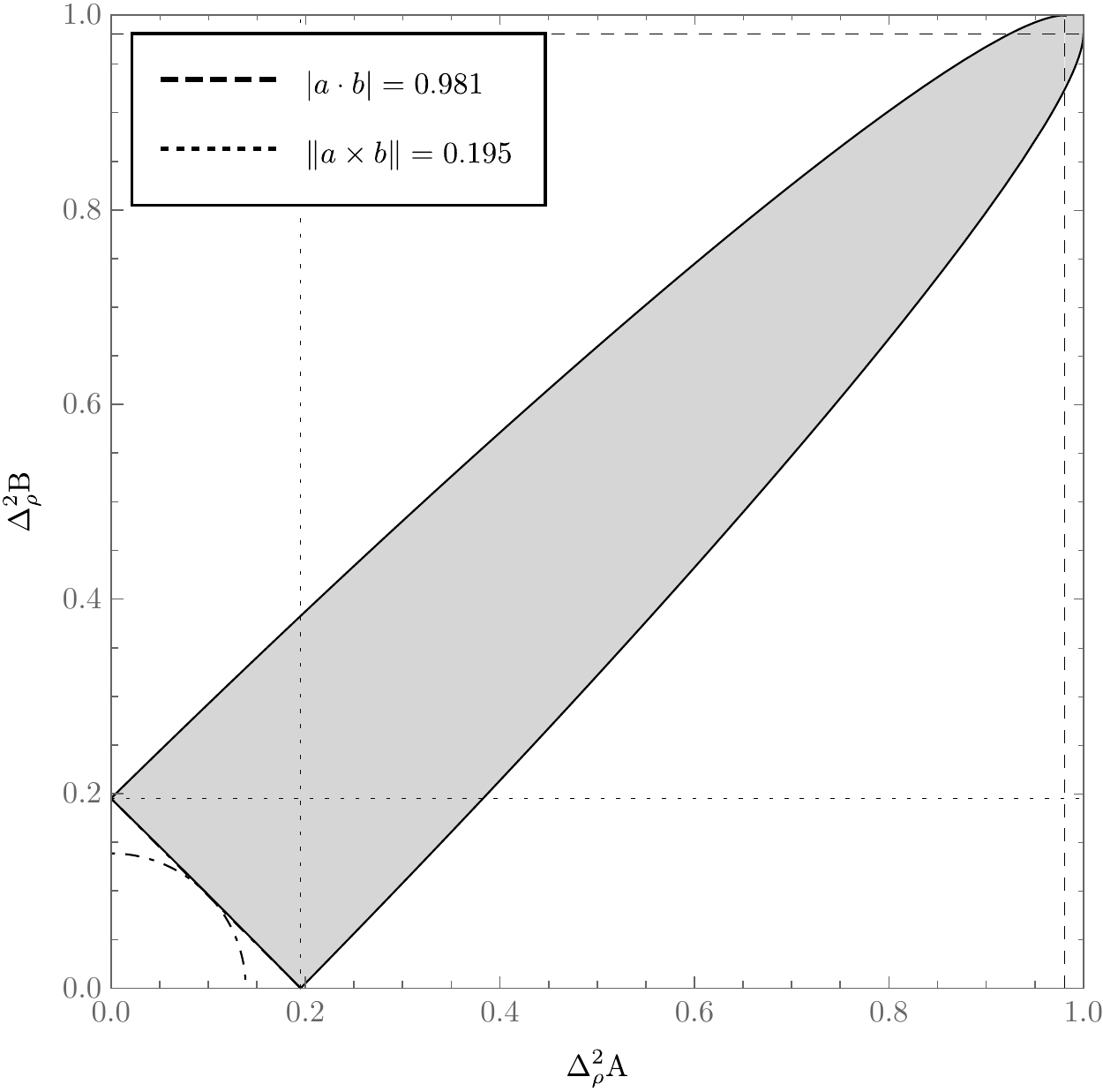}
    \caption{The uncertainty region with $\boa\cdot\bob = \cos\frac{\pi}{16}$.}
    \label{fig:err-region-pi-by-16}
  \end{subfigure}
  \begin{subfigure}[b]{0.4\textwidth}
    \includegraphics[width=\textwidth]{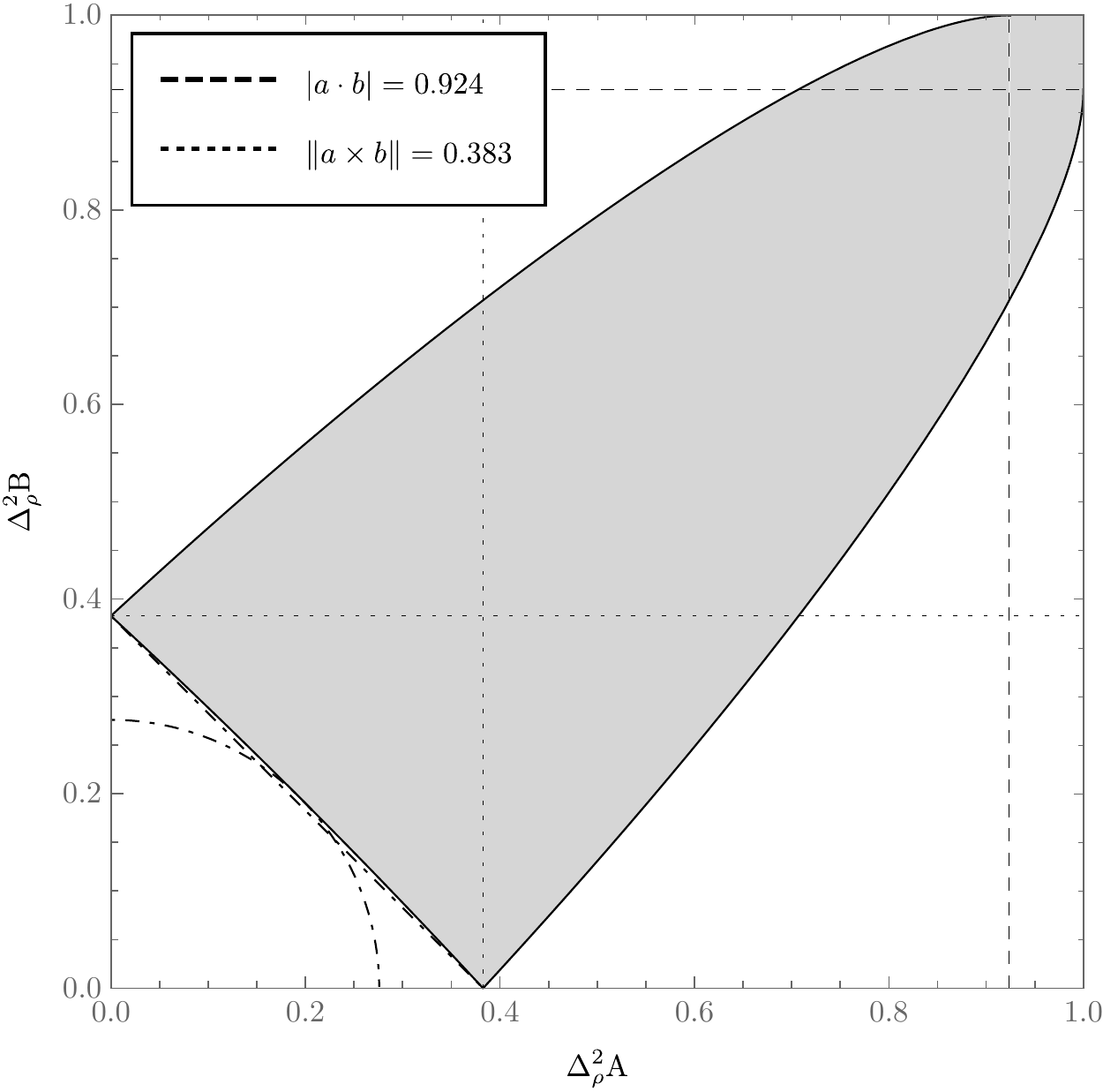}
    \caption{The uncertainty region with $\boa\cdot\bob = \cos\frac{\pi}{8}$.}
    \label{fig:err-region-pi-by-8}
  \end{subfigure}
  \begin{subfigure}[b]{0.4\textwidth}
    \includegraphics[width=\textwidth]{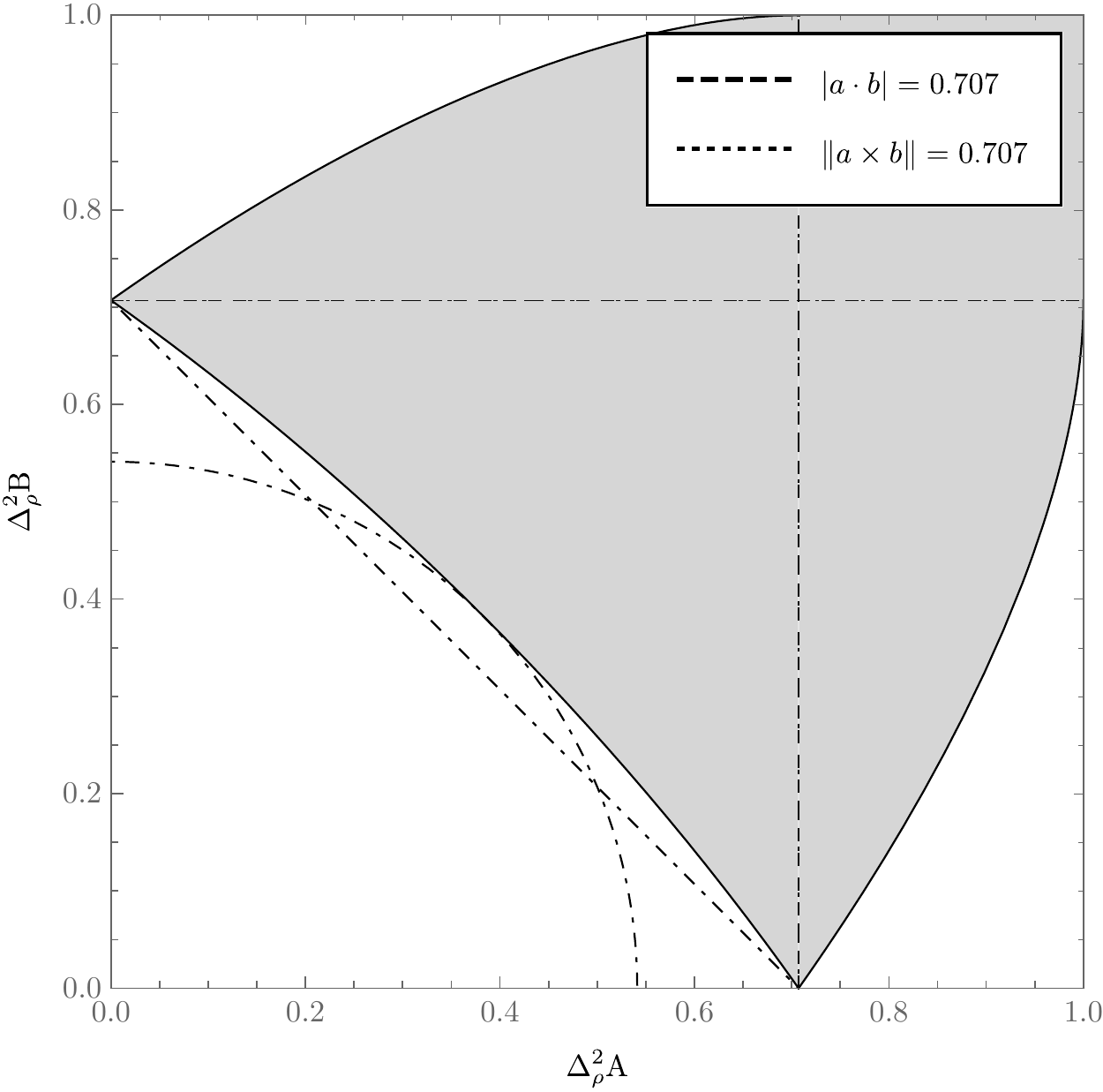}
    \caption{The uncertainty region with $\boa\cdot\bob = \cos\frac{\pi}{4}$.}
    \label{fig:err-region-pi-by-4}
  \end{subfigure}
  \begin{subfigure}[b]{0.4\textwidth}
    \includegraphics[width=\textwidth]{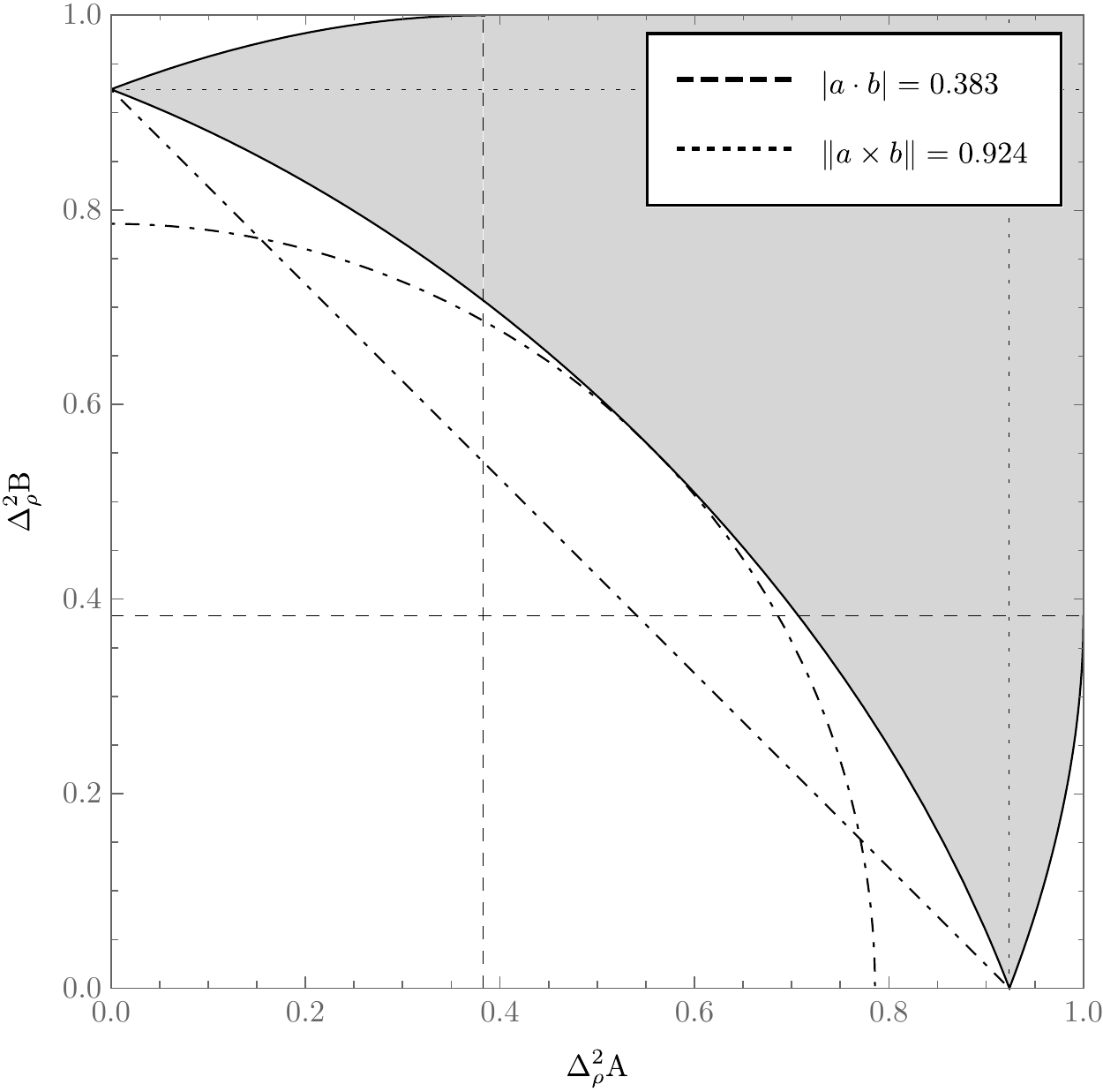}
    \caption{The uncertainty region with $\boa\cdot\bob = \cos\frac{3\pi}{8}$.}
    \label{fig:err-region-3pi-by-8}
  \end{subfigure}
  \begin{subfigure}[b]{0.4\textwidth}
    \includegraphics[width=\textwidth]{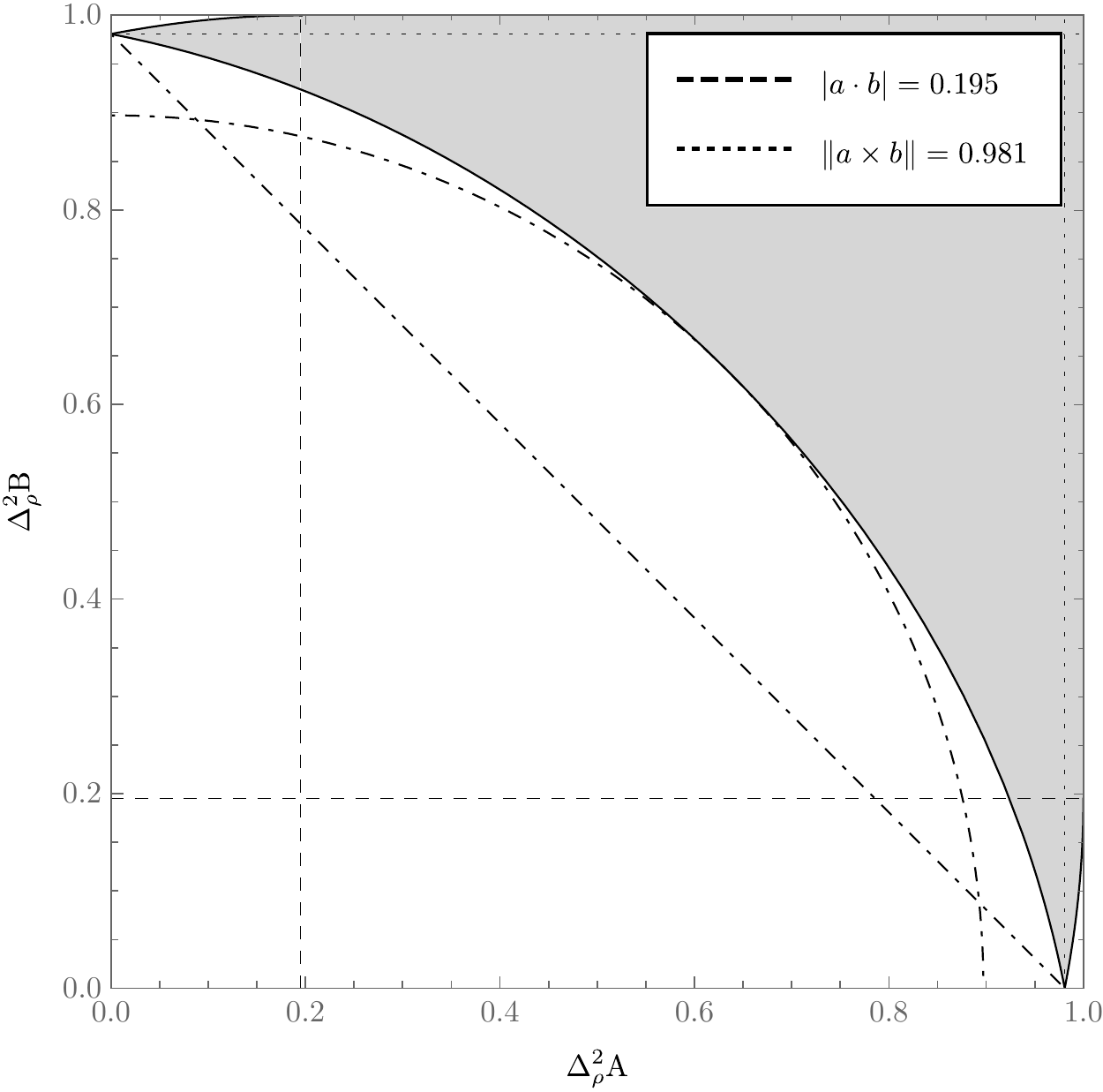}
    \caption{The uncertainty region with $\boa\cdot\bob = \cos\frac{7\pi}{16}$.}
    \label{fig:err-region-7pi-by-16}
  \end{subfigure}
  \begin{subfigure}[b]{0.4\textwidth}
    \includegraphics[width=\textwidth]{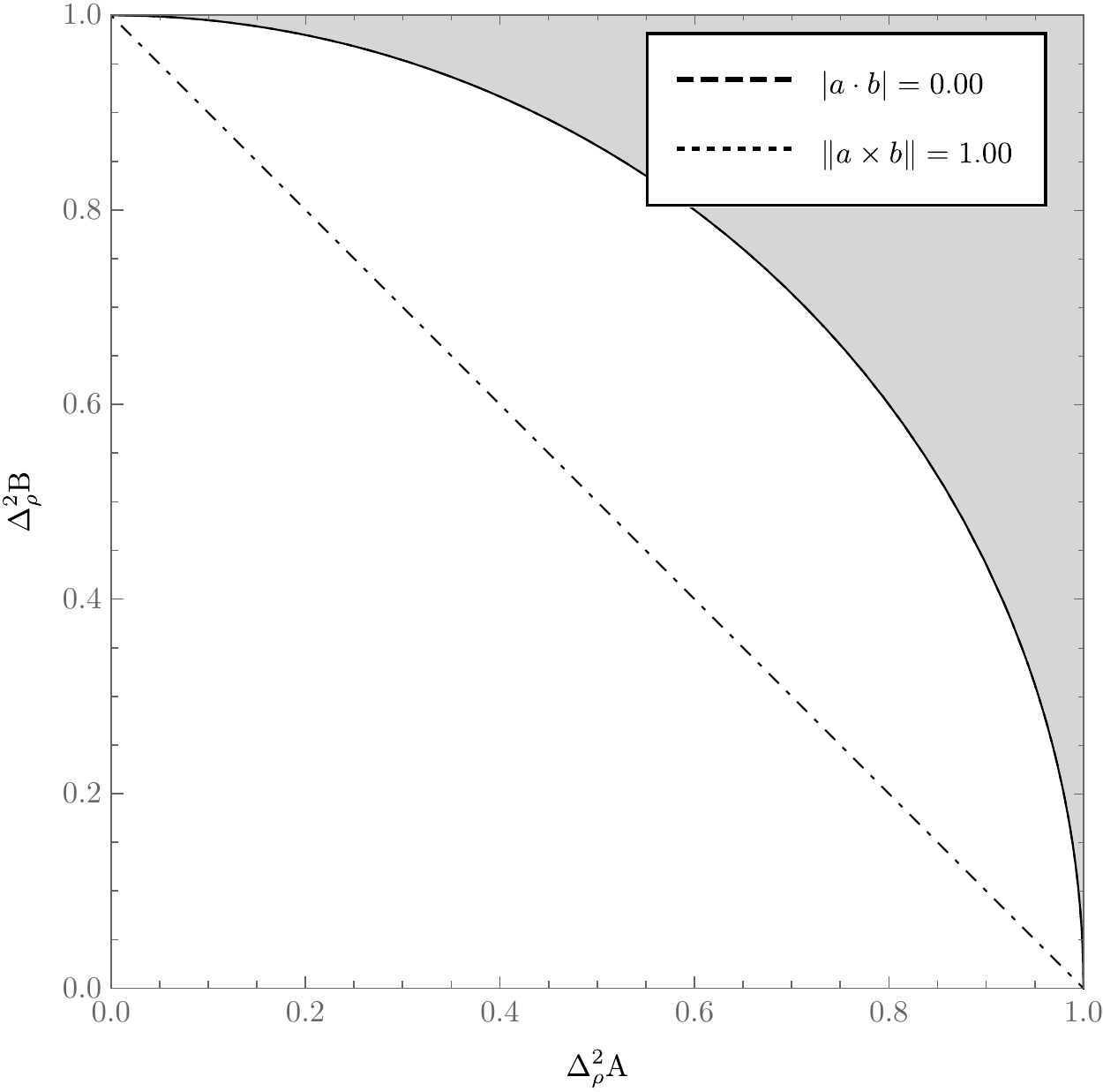}
    \caption{The uncertainty region with $\boa\cdot\bob = \cos\frac{\pi}{2}$.}
    \label{fig:err-region-pi-by-2}
  \end{subfigure}
  \caption{Plots of the uncertainty region for sharp, $\pm1$-valued qubit observables. The straight and curved dot-dashed lines are the previously known lower bounds \eqref{eqn:straight-lower-bound-blw} and \eqref{eqn:curved-lower-bound-blw}, respectively.}
  \label{fig:top-level-err-region}
\end{figure*}

\subsection{Schr\"odinger uncertainty relation}\label{sec:sch-uncertainty-relation}

We now turn to a brief analysis of the Schr\"odinger inequality, beginning with the following observation.
\begin{lem}
The identity \eqref{eq:r2-id} for unit vectors $\boa,\bob\in\mathbb{R}^3$ and any vector $\bor\in\mathbb{R}^3$can be rewritten in the following two equivalent forms:
\begin{align}
 & \left\|\boa\times\bob\right\|^2 \left\|\bor\right\|^2 = \bigl((\boa\times\bob)\cdot\bor\bigr)^2 + \bigl\|(\boa\times\bob)\times\bor\bigr\|^2\tag{17}\\
 &\iff  \left(1-(\boa\cdot\bor)^2\right)+\left(1-(\bob\cdot\bor)^2\right)+\left(\|\boa\times\bob\|^2-(\boa\times\bob\cdot\bor)^2\right)=\|\boa\times\bob\|^2\bigl(1-\|\bor\|^2\bigr)+2(1-\boa\cdot\bob\,\boa\cdot\bor\,\bob\cdot\bor)\label{eq:r2-id2}\\
 &\iff \bigl(1-(\boa\cdot\bor)^2\bigr)\bigl(1-(\bob\cdot\bor)^2\bigr)-
    \left((\boa\times\bob\cdot\bor)^2+(\boa\cdot\bob-\boa\cdot\bor\,\bob\cdot\bor)^2\right)=\|\boa\times\bob\|^2\bigl(1-\|\bor\|^2\bigr).\label{eq:r2-id3}
\end{align}
\end{lem}
\begin{proof}
Recall the identity based on the Lagrange formula for the double vector product,
  \begin{align}
    \left\|(\boa\times\bob)\times\bor\right\|^2 = \left\|\boa(\bob\cdot\bor) - \bob(\boa\cdot\bor) \right\|^2
    = (\boa\cdot\bor)^2 + (\bob\cdot\bor)^2 - 2(\boa\cdot\bob)\,(\boa\cdot\bor)\,(\bob\cdot\bor).
  \end{align}
  We use this to rewrite \eqref{eq:r2-id} as follows:
  \begin{align*}
    \|\boa\times\bob\|^2\|\bor\|^2&=(\boa\times\bob\cdot\bor)^2+(\boa\cdot\bor)^2 + (\bob\cdot\bor)^2 - 2\boa\cdot\bob\,\boa\cdot\bor\,\bob\cdot\bor)\\
                                  &=\|\boa\times\bob\|^2-\left(\|\boa\times\bob\|^2-(\boa\times\bob\cdot\bor)^2\right)+1-\left(1-(\boa\cdot\bor)^2\right)+1-\left(1-(\bob\cdot\bor)^2\right)-2 (\boa\cdot\bob)\,(\boa\cdot\bor)\,(\bob\cdot\bor)                                    
  \end{align*}
  Upon rearranging terms, we obtain \eqref{eq:r2-id2}, showing at once its equivalence with \eqref{eq:r2-id}.
  
  Next, working on the left hand side of \eqref{eq:r2-id3}, we obtain:
 \begin{align}
 &\bigl(1-(\boa\cdot\bor)^2\bigr)\bigl(1-(\bob\cdot\bor)^2\bigr)-
    \left((\boa\times\bob\cdot\bor)^2+(\boa\cdot\bob-\boa\cdot\bor\,\bob\cdot\bor)^2\right)\\
    &\quad=1-(\boa\cdot\bor)^2-(\bob\cdot\bor)^2+(\boa\cdot\bor)^2(\bob\cdot\bor)^2-(\boa\times\bob\cdot\bor)^2-(\boa\cdot\bob)^2-(\boa\cdot\bor)^2(\bob\cdot\bor)^2
    +2\boa\cdot\bob\,\boa\cdot\bor\,\bob\cdot\bor\\
    &\quad=\left(1-(\boa\cdot\bor)^2\right)+\left(1-(\bob\cdot\bor)^2\right)-1+\left(\|\boa\times\bob\|^2-(\boa\times\bob\cdot\bor)^2\right)-\left(\|\boa\times\bob\|^2+(\boa\cdot\bob)^2\right)
    +2\boa\cdot\bob\,\boa\cdot\bor\,\bob\cdot\bor\\
    &\quad=\left(1-(\boa\cdot\bor)^2\right)+\left(1-(\bob\cdot\bor)^2\right)+\left(\|\boa\times\bob\|^2-(\boa\times\bob\cdot\bor)^2\right)-2(1-\boa\cdot\bob\,\boa\cdot\bor\,\bob\cdot\bor)=: g(\boa,\bob,\bor)
 \end{align}
 Equating this with the right hand side, we see that \eqref{eq:r2-id3} implies \eqref{eq:r2-id2}.
 
 Conversely, we may use \eqref{eq:r2-id2} to see that $g(\boa,\bob,\bor)$ is actually equal $\|\boa\times\bob\|^2\bigl(1-\|\bor\|^2\bigr)$, which shows that \eqref{eq:r2-id2} implies
 \eqref{eq:r2-id3}.
\end{proof}
We recall that
for any qubit state $\rho = \frac{1}{2}\left(\operatorname{I} + \bor\cdot\bosig\right)$ we have
\begin{align}
  \varr{\opa} &= 1-\left(\boa\cdot\bor\right)^2,\quad  \varr{\opb} = 1-\left(\bob\cdot\bor\right)^2,\\
  \bigl|\expr{\com{\opa}{\opb}}\bigr| &= 2\left|\left(\boa\times\bob\right)\cdot\bor\right|,\\
  \expr{\acom{\opa}{\opb}} &= 2(\boa\cdot\bob),\quad  \expr{\opa} = \boa\cdot\bor,\quad \expr{\opb} = \bob\cdot\bor.
\end{align}
Further, we note that the variance of the observable $\opc=\boa\times\bob\cdot\bosig$ is $\varr\opc=\|\boa\times\bob\|^2-(\boa\times\bob\cdot\bor)^2$.
This can be used to translate the above identities into two equivalent forms of uncertainty equations.

\begin{thm}\label{thm:3ur}\rm
  The observables $\opa=\boa\cdot\bosig$, $\opb=\bob\cdot\bosig$, and $\opc=\boa\times\bob\cdot\bosig$ obey the following equivalent uncertainty equations for all states $\rho=\frac12(I+\bor\cdot\bosig)$:
  \begin{equation}\label{eq:triple-ur}
    \varr\opa+\varr\opb+\varr\opc=\|\boa\times\bob\|^2\bigl(1-\|\bor\|^2\bigr)+2(1-\boa\cdot\bob\,\boa\cdot\bor\,\bob\cdot\bor),
  \end{equation}
\begin{equation}\label{eq:S-eq}
  \varr\opa\,\varr\opb-\left[\frac14 \bigl|\expr{\com{\opa}{\opb}}\bigr|^2+\frac14 \bigl(\expr{\acom{\opa}{\opb}}-2\expr{\opa}\expr{\opb}\bigr)^2\right]=\|\boa\times\bob\|^2\bigl(1-\|\bor\|^2\bigr).
\end{equation}
\end{thm}

This yields, in particular, the Schr\"odinger inequality \eqref{eq:SUR}.

The Schr\"odinger inequality does not have the form of a  state-independent uncertainty relation, except in the case $\boa\cdot\bob=0$ (treated in Subsection \ref{sec:xyz}).
Nevertheless, it does provide a specification of the lower boundary of the uncertainty relation. The upper boundary is obtained  by appliciation of the full equation \eqref{eq:S-eq}.
\begin{cor}\rm
  The upper and lower boundary value of each vertical segment  $\bigl\{(\sdev{\opa},\sdevr{\opb})\,|\, \rho\in\si(\sdev{\opa})\bigr\}$ of the uncertainty region for $\opa=\boa\cdot\bosig,\opb=\bob\cdot\bosig$ is determined by the Schr\"odinger bound
  \begin{equation}
    S(\opa,\opb,\rho)=\frac14 \bigl|\expr{\com{\opa}{\opb}}\bigr|^2+\frac14 \bigl(\expr{\acom{\opa}{\opb}}-2\expr{\opa}\expr{\opb}\bigr)^2
  \end{equation}
  as follows: 
  \begin{alignat}{2} 
    \varmin{\opb}&=\min&&\left\{\frac{S(\opa,\opb,\rho)}{\var{\opa}}\,\middle|\,\rho\in\si(\sdev{\opa})\right\},\\
    \varmax{\opb}&=\max&&\left\{\frac{S(\opa,\opb,\rho)}{\var{\opa}}\,\middle|\,\rho\in\si(\sdev{\opa})\right\}.
  \end{alignat}
\end{cor}
\begin{proof}
  This is a direct consequence of Eq.~\eqref{eq:S-eq} and the fact that the maximizing and minimizing states can be chosen to be pure.
\end{proof}
Thus we find that the strengthening \eqref{eq:S-eq} of the Schr\"odinger inequality into an equation determines the uncertainty region for $\boa\cdot\bosig,\ \bob\cdot\bosig$. However, the Schr\"odinger inequality itself gives the lower bound for $\sdevr{\opb}$ given $\sdev{\opa}$, and similarly the lower bound for $\sdev{\opa}$ given $\sdevr{\opb}$. Since the uncertainty region is symmetric under reflection on the axis $\sdev{\opa}=\sdev{\opb}$, the minimal boundaries for the two uncertainties together, obtained by the Schr\"odinger inequality alone, determine the uncertainty region.

An unexpected feature becomes apparent in the case of minimal uncertainty. Note that one may always move the vector $\bor$ into the plane spanned by $\boa$ and $\bob$ without changing the variances $\varr{\opa}$ and $\varr{\opb}$. Since $\bor$ is then perpendicular to $\boa\times\bob$ the ``commutator term" $\left(\left(\boa\times\bob\right)\cdot\bor\right)^2$ in the uncertainty relation \eqref{eq:S-eq} is zero for all of these vectors. Hence the lower uncertainty bound (which is always assumed on unit vectors, so that the above corollary remains applicable) is a feature purely of the anti-commutator term. 
This term is analogous in form to the classical covariance; however, in the quantum context, this interpretation only applies where the measurements are compatible and thus have physical joint probabilities.

\section{Qutrit uncertainty}
\label{sec:qutrit}
\subsection{Extended qubit observables}
\label{sec:qubit-extended-uncertainty}
A natural continuation of the qubit example is provided by extending the general, sharp, $\pm 1$-valued qubit observables $\boa\cdot\bosig$ and $\bob\cdot\bosig$ into a third dimension
\begin{alignat}{3}
  \label{eqn:extended-qubit-obs-defn}
  \opa &= &(\boa\cdot\bosig) \oplus 0 = &\begin{pmatrix}\boa\cdot\bosig & 0\\ 0 &0\\\end{pmatrix}\\
  \opb &= &(\bob\cdot\bosig) \oplus 0 = &\begin{pmatrix}\bob\cdot\bosig & 0\\ 0 &0\\\end{pmatrix},
\end{alignat}
where $\boa$ and $\bob$ are normalised, and $\bosig$ is the usual vector of qubit Pauli matrices. It is easily verified that given any qutrit density matrix we can attain the same variance pairs $\var{\opa}, \var{\opb}$ with a density matrix of the form
\begin{align}
  \rho =\frac{w}{2}\left(\opi_2 + \bor\cdot\bosig\right) \oplus (1-w)=\begin{pmatrix}\frac{w}{2}\left(\opi_2 + \bor\cdot\bosig\right) & 0 \\ 0& (1-w)\end{pmatrix},
\end{align}
where $\frac{1}{2}\left(\opi_2 + \bor\cdot\bosig\right)$ is a qubit density matrix, and $w$ is a real parameter between $0$ and $1$ (inclusive). We can compute the variances of $\opa$ and $\opb$ for a state of this form directly from the definition
\begin{align}
  \varr{\opa} &= w - w^2 (\boa\cdot\bor)^2\\
  \varr{\opb} &= w - w^2 (\bob\cdot\bor)^2.
\end{align}
Unfortunately an analytical description of the uncertainty region does not seem to be forthcoming for the case of general $\boa$ and $\bob$, although numerical approximations to the boundary curve may readily be computed. We therefore focus our attention on the case $\boa\cdot\bob = 0$. We note that projecting a vector onto the plane spanned by $\boa$ and $\bob$ leaves both of variances unchanged so, without loss of generality, set 
\begin{align}
  \bor = r_a \boa + r_b \bob,
\end{align}
subject to
\begin{align}
  r_a^2 + r_b^2 \leq 1.
\end{align}
At a fixed $w$ the minimum for $\varr{\opb}$ will be attained by making $(\bob\cdot\bor)^2$ as large as possible; we therefore set $r_b^2 = 1-r_a^2$. We also see that for $X\in [0,1]$ the equation $X = \varr{\opa}$ enforces a relation between $w$ and $r_a^2$: 
\begin{align}
  w_\pm = \frac{1\pm\sqrt{1 -4 X r_a^2}}{2r_a^2}.
\end{align}
Since $w$ is required to be real for $\rho$ to be a valid state, we need $r_a^2 \leq \frac{1}{4X}$; in addition $w$ must be in the range $[0,1]$. Note that  $w_+\ge w_-\ge 0$, so that $w_-$ leads to a valid state whenever $w_+$ does. Now, $w_+\le 1$ is equivalent to having both $r_a^2\ge\frac12$ and $r_a^2\ge 1-X$. Denoting
\begin{align}
  Y_\pm = w_\pm - w_\pm^2 (1-r_a^2),
\end{align}
we have that
\begin{align}
  Y_+ - Y_- =\bigl(2r_a^2-1\bigr)\frac{\sqrt{1-4 X r_a^2}}{r_a^4}.
\end{align}
Hence, wherever $w_+$ leads to a valid quantum state, $w_-$ gives a lower $\varr{\opb}$, and so we can focus on $w_-, Y_-$.
The requirement  $w_-\le 1$ is satisfied if and only if $r_a^2 \leq 1-X$ whenever $r_a^2 < \frac{1}{2}$. 
We now note that $w_-(r_a^2)$  always gives a valid solution when $r_a^2 = 0$,
\begin{align}
  w_-(0) &= \lim_{r_a^2 \to 0}  \frac{1 - \sqrt{1 -4 X r_a^2}}{2r_a^2} = X,\\
  Y_-(0) &= X(1-X),
\end{align}
It is easily verified that $w_-(r_a^2)\equiv w_-(u)>X$ whenever $u=r_a^2>0$; this entails that the derivative 
\begin{equation}\label{eq:w_->X}
w_-'(u)=\frac{w_-(u)-X}{u\sqrt{1-4Xu}}>0\quad\text{for }u>0.
\end{equation}
We then differentiate $Y_-(u)$
\begin{align}
  Y_-^\prime(u) = \frac{2(w_--X)(1-w_-)}{u\sqrt{1-4Xu}}\ge 0,
\end{align}
so that $Y_-(r_a^2) - Y_-(0)\ge 0$ always.
Hence we  take $r_a^2 = 0$ to find the minimum $\varmin{\opb}=X(1-X)$ at a fixed $\var{\opa}=X$. The lower boundary of the uncertainty region is therefore given by the curve $\sdev{\opb} = \sdev{\opa}\sqrt{1-\var{\opa}}$. Since the region is symmetric under reflection on the axis $\sdev\opa=\sdev\opb$ and in the present case must contain the uncertainty region for orthogonal qubit observables, it is given by the set
\begin{align}
  {\rm PUR}_\Delta(\opa,\opb) = \left\{(\sdev{\opa},\sdev{\opb})\in [0,1]\times[0,1]\, \middle|\, \sdev{\opb} \geq \sdev{\opa}\sqrt{1-\var{\opa}} \text{ and } \sdev{\opa} \geq \sdev{\opb}\sqrt{1-\var{\opb}}  \right\},
\end{align}
shown in Figure \ref{fig:extended-qubit-uncertainty-region}.

\begin{center}
  \begin{figure}[ht]
    \includegraphics[width=.5\textwidth]{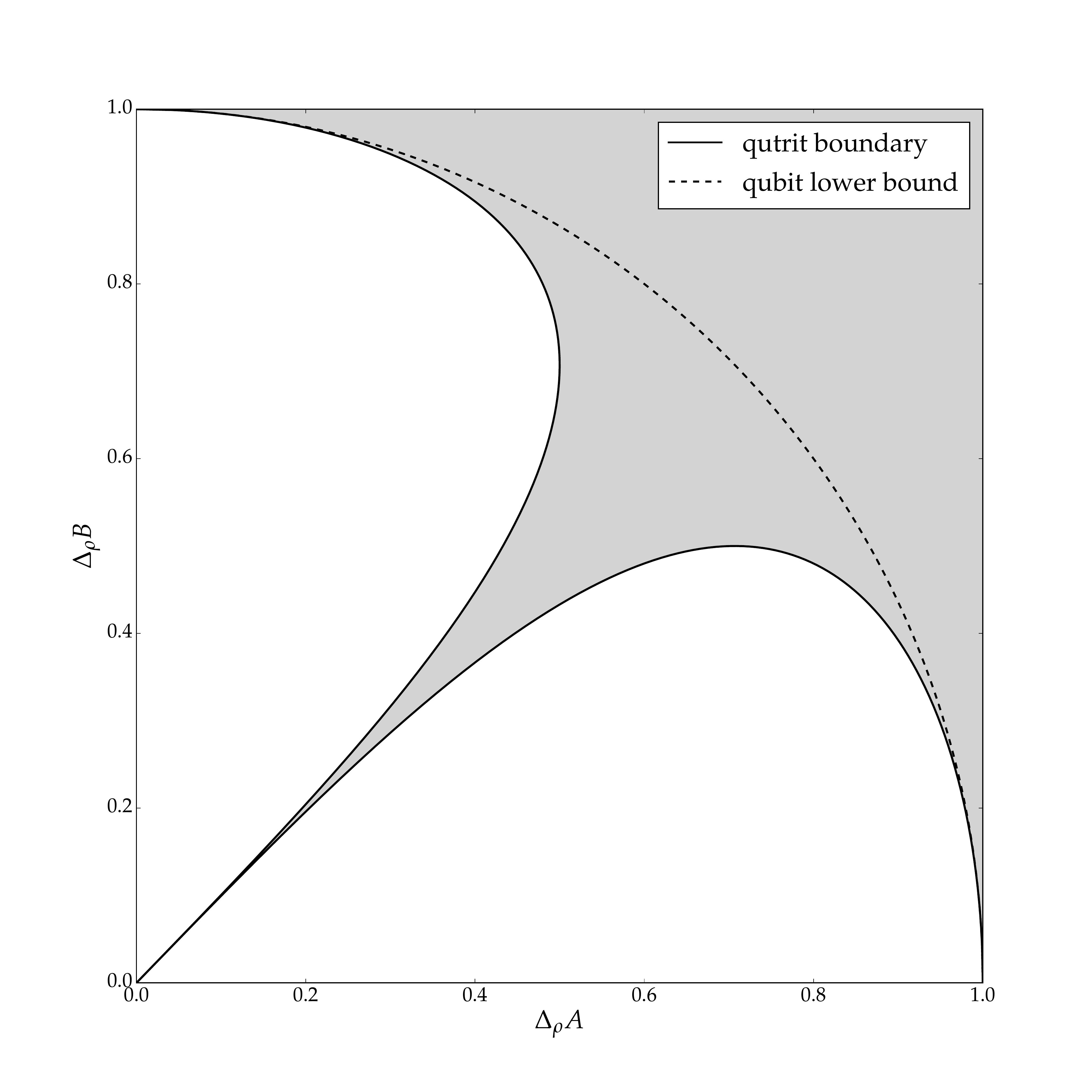}
    \caption{The uncertainty region for the qutrit observables defined in equation \eqref{eqn:extended-qubit-obs-defn}. The dashed line indicates the lower boundary of the set of standard deviation pairs achievable by states of the form $\rho_2 \oplus 0$, where $\rho_2$ is a qubit density matrix. The points $(0,1)$, $(0,0)$ and $(1,0)$ are attained by the states $\frac{1}{2}(\opi_2+\bob\cdot\bosig)\oplus 0$, $0\oplus 1$, and $\frac{1}{2}(\opi_2+\boa\cdot\bosig)\oplus 0$ respectively.}
    \label{fig:extended-qubit-uncertainty-region}
  \end{figure}
\end{center}

\subsection{``Gell-Mann'' observables}
An interesting counterpoint to section \ref{sec:qubit-uncertainty} is provided by the case of quantum observables on a three dimensional Hilbert space. Here it is possible to show, by counterexample, that the Schr\"odinger uncertainty relation is not sufficient to define the exact uncertainty region. We expect that the same will hold true for all finite dimensions greater than two.
For our counterexample we choose the observables to be two of the Gell-Mann matrices, and let $\rho$ be an arbitrary, Hermitian, positive-semi-definite three by three matrix of trace 1.
\begin{equation}
  \label{eq:qutrit-observables-def}
  \opa = \begin{pmatrix}
    1 & 0 & 0\\
    0 & -1 & 0\\
    0 & 0 & 0\\
  \end{pmatrix} \quad
  \opb = \begin{pmatrix}
    0 & 0 & 1\\
    0 & 0 & 0\\
    1 & 0 & 0\\
  \end{pmatrix} \quad
  \rho =\begin{pmatrix}
    \rho_{11} & \rho_{12} & \rho_{13}\\
    \rho_{12}^* & \rho_{22} & \rho_{23}\\
    \rho_{13}^* & \rho_{23}^* & \rho_{33}
  \end{pmatrix}
\end{equation}
Then
\begin{equation}
  \begin{aligned}
    \opa^2 &= \begin{pmatrix}
      1 & 0 & 0\\
      0 & 1 & 0\\
      0 & 0 & 0\\
    \end{pmatrix}\qquad
    &\com{\opa}{\opb} = \begin{pmatrix}
      0 & 0 & 1\\
      0 & 0 & 0\\
      -1 & 0 & 0\\
    \end{pmatrix}\\
    \opb^2 &= \begin{pmatrix}
      1 & 0 & 0\\
      0 & 0 & 0\\
      0 & 0 & 1\\
    \end{pmatrix}
    &\acom{\opa}{\opb} = \begin{pmatrix}
      0 & 0 & 1\\
      0 & 0 & 0\\
      1 & 0 & 0\\
    \end{pmatrix}
  \end{aligned},
\end{equation}
\begin{align}
  \expr{\opa} &= \rho_{11} - \rho_{22}\\
  \expr{\opb} &= \rho_{13} + \rho_{13}^* = 2\operatorname{Re}{\rho_{13} }\\
  \expr{\opa^2} &= \rho_{11} + \rho_{22}\\
  \expr{\opb^2} &= \rho_{11} + \rho_{33} = 1 - \rho_{22}\\
  \expr{\com{\opa}{\opb}} &= \rho_{13} - \rho_{13}^* = 2\operatorname{Im}{\rho_{13} }\\
  \expr{\acom{\opa}{\opb}} &= \rho_{13} + \rho_{13}^* = \expr{\opb}\\
  \varr{\opa} &= \rho_{11} + \rho_{22} - (\rho_{11} - \rho_{22})^2\\
  \varr{\opb} &= \rho_{11} + \rho_{33} - 4(\operatorname{Re}{\rho_{13}})^2.
\end{align}
We can set $\rho_{12}$ and $\rho_{23}$ equal to zero without changing the uncertainties or the Schr\"odinger relation at all. Note that the new matrix we obtain by this procedure is positive semi-definite and trace $1$ if the original was. We can, therefore, explore the entire uncertainty region using states of the form
\begin{equation} \label{eq:qutrit-specialised}
  \rho = \begin{pmatrix}
    \rho_{11} & 0 & \rho_{13}\\
    0 & \rho_{22} & 0\\
    \rho_{13}^* & 0 & \rho_{33}
  \end{pmatrix}.
\end{equation}

By elementary methods (differentiating, finding local extrema and comparing them) we can find the minimum and maximum values of $\varr{\opb}$ as a function of $\varr{\opa}$. Because of the way the various constraints change with $\varr{\opa}$ the functional form of the minima and maxima also change. In all there are ten distinct bounding curve segments, given in equation \eqref{eqn:gellmann-bounding-curves} and shown in Figure \ref{fig:qutrit-uncertainty-region}. We give a derivation of these curves in Appendix \ref{app:gellmann-ur}.
\begin{center}
  \begin{figure*}[ht]
    \includegraphics[width=.8\textwidth]{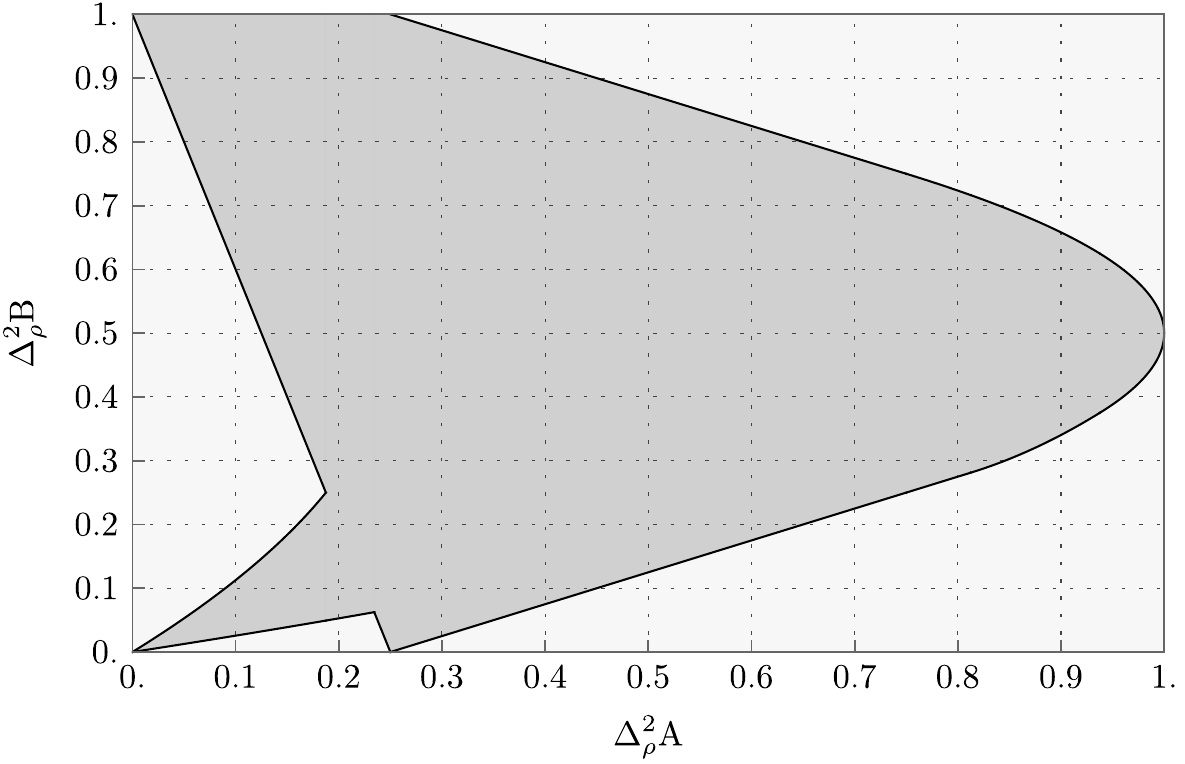}
    \caption{The uncertainty region for the qutrit observables defined in equation \eqref{eq:qutrit-observables-def}. The region contained in the solid curves is the allowed uncertainty region.}
    \label{fig:qutrit-uncertainty-region}
  \end{figure*}
\end{center}
Similar to the qubit case, the uncertainty region contains nontrivial upper bounds, and it is not of a simple convex shape; however, there are fundamental differences. The region shown in Fig.~\ref{fig:qutrit-uncertainty-region} does touch and include the origin (0,0), reflecting the fact that the two observables have a common eigenstate. The shape of the region is also quite asymmetrical; in particular, it is not possible for both uncertainties to get large simultaneously. It is possible that these features can be connected to trade-off relations involving other observables, as we indicated in the qubit case. However, this may require the acquisition of a host of further case studies.
The Schr\"odinger relation does not entail the lower bound of the uncertainty region in this case. 
We show this by determining the maximum value in the interval of possible values of the Schr\"odinger bound,
$\{S(\opa,\opb,\rho)\,|\, \rho\in\si(\sdev{\opa})\}$, and we find indeed that for some range of values of $\sdev{\opa}$,
\begin{equation}\label{eq:sch-uncertainty-region}
  \varmin{\opb} > \max\left\{\frac{1}{4\varr{\opa}} \left(\bigl| \expr{\com{\opa}{\opb}} \bigr|^2 
      + 
      \bigl(\expr{\acom{\opa}{\opb}}-2\expr{\opa}\expr{\opb}\bigr)^2\right)\right\}=\max\left\{\frac{S(\opa,\opb,\rho)}{\var{\opa}}\right\}.
\end{equation}
To verify this we first solve the equation
\begin{align}
  x &= \varr{\opa}\\
    &= 1 - \rho_{33} - \left(2\rho_{11} + \rho_{33} - 1\right)^2\\
  \implies \rho_{33}^\pm &= \frac{1}{2}\left(1 - 4\rho_{11} \pm\sqrt{1+8\rho_{11}-4x}\right).
\end{align}
We then note that in the range $x\in\left[\frac{3}{4} ,1\right]$ only the $\rho_{33}^+$ solution with $\rho_{11} \in \left[\frac{1}{2} - \frac{\sqrt{1-x}}{2},\frac{1}{2} + \frac{\sqrt{1-x}}{2}\right]=:\ival$ leads to $\rho$ being a valid state (positive and trace $1$). We therefore seek
\begin{align}
  f(x) &\coloneqq \frac{1}{x}\max\left\{\Im(\rho_{13})^2 +\Re(\rho_{13})^2\left(6-8\rho_{11}-4\rho_{33}^+\right)^2\middle|\, \rho_{11}\in \ival,\,\left\lvert\rho_{13}\right\rvert^2 \leq \rho_{11}\rho_{33}^+\right\}\\
       &= \frac{1}{x}\max\left\{\Im(\rho_{13})^2 +\Re(\rho_{13})^2\left(4-2\sqrt{1+8\rho_{11}-4x}\right)^2\middle|\, \rho_{11}\in\ival,\,\left\lvert\rho_{13}\right\rvert^2 \leq \rho_{11}\rho_{33}^+\right\}\\
       &= \frac{1}{x}\max\left\{\left(\lambda +(1-\lambda)\left(4-2\sqrt{1+8\rho_{11}-4x}\right)^2\right)\frac{\rho_{11}}{2}\left(1-4\rho_{11}+\sqrt{1+8\rho_{11}-4x}\right)\middle|\, \rho_{11}\in\ival,\,\lambda\in[0,1]\right\}.
\end{align}
For ease of exposition we here restrict our attention to $x = 1$, in which case only $\rho_{11} = \frac{1}{2}$ leads to a valid quantum state. We can therefore directly compute $f(1) = 0$, and note that as the function is continuous there is an interval where the Schr\"odinger inequality is too weak to completely describe the uncertainty region.  

\section{Conclusion} \label{sec:conclusion}
In this paper we have introduced the notion of the uncertainty region for a pair (or a finite collection) of quantum observables, and provided a range of examples illustrating the concept. In contrast to the well-known uncertainty relations, we observed that an uncertainty region is most appropriately described by a {\em state-independent} form of relation that describes, in particular, its boundary.

We have given a geometrical derivation of the exact uncertainty region for an  arbitrary pair of $\pm1$-valued qubit observables, in the explicit form of a state independent uncertainty relation.  When the observables $\opa,\opb$ have non-orthogonal Bloch vectors $\boa,\bob$, we found non-trivial upper bounds for the variance $\varr{\opb}$ as a function of $\varr{\opa}$, and showed that this may be understood in terms of the uncertainty trade-off between $\opa$ and another observable $\opb'$ (whose Bloch vector $\bob'$ is in the plane of $\boa,\bob$ and perpendicular to $\bob$): the observables $\opb,\opb'$ obey the uncertainty relation $\var{\opb}+\var{\opb'}\ge 1$, and then the minimum value of $\sdev{\opb'}$ given $\sdev{\opa}$ dictates the maximum value of $\sdev{\opb}$.

We have seen that the Schr\"odinger inequality determines the uncertainty region in the qubit case, despite the fact that it is only a state-independent inequality in the case where $\boa\perp\bob$. This is essentially due to the fact that satisfaction of this inequality is equivalent to the positivity condition for states.

Finally we described the uncertainty region for two pairs of qutrit observables, which provide illustrations of the often non-trivial shape of an uncertainty region. The pairs of observables studied here do have a common eigenstate and consequently the uncertainty region is allowed to touch and include the point (0,0). The last example also demonstrates the fact that the Schr\"odinger relation cannot, in general, determine the lower boundary (and certainly not the upper boundary) of the uncertainty region in dimensions higher than two.

The examples studied here reinforce the qualitative understanding of  the uncertainty principle  as the statement that the incompatibility (non-commutativity) of a pair of observables generally enforces a state-independent lower bound to their uncertainty region. Where incompatible observables do have joint eigenstates, allowing the uncertainty region to include the origin, one must still expect that parts of some neighbourhood of (0,0) will remain excluded from the uncertainty region. 

The general theory of the structure of uncertainty regions is still unknown. It seems likely that an expanding library of case studies, like those described above, will help point the way for future investigations of this theory. A notable feature of these investigations is how rapidly the computations become more difficult as the Hilbert space dimension increases, for example attempting to generalise the results of sections \ref{sec:qubit-uncertainty} and \ref{sec:qubit-extended-uncertainty} to the case of extended qubit observables with non-orthogonal Bloch vectors requires computing the roots of fifth order polynomials. One avenue for further investigation could be the use of numerical methods in the analysis; since the variance is quadratic in the state the problem may be reduced to polynomial root finding, which may be efficiently solved using well known numerical techniques.
To conclude, we expect that much can be learned about the uncertainty principle through the study of uncertainty regions, and hope our investigation will encourage some readers to undertake further case studies.

\section*{Acknowledgements}

We would like to thank Stefan Weigert and Roger Colbeck for helpful and stimulating discussions, as well as for feedback on previous drafts. Thomas Bullock kindly provided useful comments on the first arXiv version. Oliver Reardon-Smith gratefully acknowledges the support of the Engineering and Physical Sciences Research Council, as well as that of the University of York Department of Mathematics.

\section*{Author contributions}

Both authors contributed extensively to the work. Sadly, Paul Busch passed away in June 2018, when this manuscript was almost finalised. Since then it has been completed by Oliver Reardon-Smith under the supervision of Roger Colbeck.

\appendix

\section{Uncertainty region for Gell-Mann observables}
\label{app:gellmann-ur}
Given
\begin{align}
  \opa = \begin{pmatrix}
    1 & 0 & 0\\
    0 & -1 & 0\\
    0 & 0 & 0\\
  \end{pmatrix} \quad
  \opb = \begin{pmatrix}
    0 & 0 & 1\\
    0 & 0 & 0\\
    1 & 0 & 0\\
  \end{pmatrix} \quad
  \rho =\begin{pmatrix}
    \rho_{11} & 0 & \rho_{13}\\
    0 & 1-\rho_{11} - \rho_{33} & 0\\
    \rho_{13}^* & 0 & \rho_{33}
  \end{pmatrix},
\end{align}
we can solve
\begin{align}
  x &= \varr{\opa}\\
    &= 1 - \rho_{33} - \left(2\rho_{11} + \rho_{33} - 1\right)^2,
\end{align}
giving
\begin{align}
  \rho_{33}^\pm = \frac{1}{2}\left(1-4\rho_{11} \pm\sqrt{1+8\rho_{11}-4x}\right).
\end{align}
The positivity of $\rho$ constrains the choice of $\rho_{11}$ values in each case. If
\begin{align}
  \rho^\pm = \begin{pmatrix}
    \rho_{11} & 0 & \rho_{13}\\
    0 & 1-\rho_{11} - \rho_{33}^\pm & 0\\
    \rho_{13}^* & 0 & \rho_{33}^\pm
  \end{pmatrix}
\end{align}
and $0\leq\rho_{13}\leq\sqrt{\rho_{11}\rho_{33}}$ then
\begin{align}
  \label{eqn:rho-plus-constraints}
  \rho^+ \geq 0 &\iff \begin{cases}\left[ 0\leq \rho_{11}\leq\frac{1}{2}(1-\sqrt{1-4x})\right] \text{ or } \left[\frac{1}{2}(1+\sqrt{1-4x})\leq\rho_{11}\leq\frac{1}{2}(1+\sqrt{1-x})\right],  &0\leq x\leq\frac{1}{4}\\
    \frac{1}{8}\left(4x-1\right)\leq \rho_{11}\leq \frac{1}{2}\left(1+\sqrt{1-x}\right), & \frac{1}{4}\leq x\leq \frac{3}{4}\\
    \frac{1}{2}\left(1-\sqrt{1-x}\right)\leq \rho_{11} \leq\frac{1}{2}(1+\sqrt{1-x}, & \frac{3}{4}\leq x\leq 1
  \end{cases}\\ 
  \rho^- \geq 0 &\iff \begin{cases}0\leq\rho_{11}\leq\frac{1}{2}(1-\sqrt{1-x}), &0\leq x\leq\frac{1}{4}\\
    \frac{1}{8}\left(4x-1\right)\leq\rho_{11}\leq\frac{1}{2}(1-\sqrt{1-x}), & \frac{1}{4}\leq x\leq \frac{3}{4}\\
    \text{no valid solution}, & \frac{3}{4}\leq x\leq 1.
  \end{cases}\label{eqn:rho-minus-constraints}
\end{align} 
The constraints on $\rho_{13}$ imply that $0\leq\left(\Re{\rho_{13}}\right)^2\leq\rho_{11}\rho_{33}^\pm$. Obviously $\var[\rho^\pm]{\opb}$ will be minimised by a $\rho^\pm$ with $\left(\Re{\rho_{13}}\right)^2 = \rho_{11}\rho_{33}^\pm$ and maximised when $\left(\Re{\rho_{13}}\right)^2 = 0$.
\begin{align}
  \var[\rho^\pm]{\opb} = \rho_{11} +\rho_{33}^\pm - 4\lambda\rho_{11}\rho_{33}^\pm.
\end{align}
For a fixed $x$ the local minima and maxima will either be where the inequalities above are saturated or where the derivative of $\var[\rho^\pm]{\opa}$ with respect to $\rho_{11}$ (considering $\rho_{33}^\pm$ as a function of $\rho_{11}$) is zero. 

\subsection{Exploring minima}
Here we consider the case $\left(\Re{\rho_{13}}\right)^2 = \rho_{11}\rho_{33}^\pm$. In this case
\begin{align}
  \var[\rho^\pm]{\opb} &= \rho_{11} + \rho_{33}^\pm -4\rho_{11}\rho_{33}^{\pm}\\
                       &= \frac{1}{2}\left(1 -6\rho_{11} + 16\rho_{11}^2 \pm (1-4\rho_{11})\sqrt{1+8\rho_{11}-4x}\right)\\
  \diff{\left(\var[\rho^\pm]{\opb}\right)}{\rho_{11}} &= -3 + 16\rho_{11} \mp 2\sqrt{1+8\rho_{11}-4x}\pm \frac{2-8\rho_{11}}{\sqrt{1+8\rho_{11} -4x}}\\
  \diff{\left(\var[\rho^\pm]{\opb}\right)}{\rho_{11}} &= 0 \iff (3 - 16\rho_{11})\sqrt{1+8\rho_{11} -4x} =\pm \left(8x-24\rho_{11}\right).
                                                        \label{eqn:deriv-y-rhomin-is-zero}
\end{align}
The solutions to this equation obey a cubic equation
\begin{align}
  (3 - 16\rho_{11})^2(1+8\rho_{11} -4x) &=\left(8x-24\rho_{11}\right)^2\\
  0&= (32 \rho_{11}-16 x+3) (8 \rho_{11} (8 \rho_{11}-5)+4 x+3),
\end{align}
with solutions
\begin{align}
  \rho_{11}^\pm &= \frac{1}{16} \left(5\pm\sqrt{13-16 x}\right)\\
  \rho_{11}^0 &= \frac{1}{32} \left(16 x-3 \right).
\end{align}
Substituting these back into \eqref{eqn:deriv-y-rhomin-is-zero} we see that $\rho_{11}^0$ and $\rho_{11}^+$ are solutions wherever they give valid quantum states, but $\rho_{11}^-$ is only a solution if $x=\frac{9}{16}$ or $\frac{3}{4}\leq x$. Comparing the solutions with the restrictions \eqref{eqn:rho-plus-constraints} we get the following solutions for $\rho^+$, and no solutions for $\rho^-$
\begin{subequations}
\label{eqn:rho-plus-zero-deriv-solns}
\begin{align}
  \rho_{11} = \frac{1}{32} \left(16 x-3 \right) &\text{ on } x\in\left[\frac{3}{16}, \frac{15}{16}\right]\\
  \rho_{11} = \frac{1}{16} \left(5+\sqrt{13 - 16x} \right) &\text{ on } x\in\left[\frac{9}{100},\frac{13}{16}\right]\\
  \rho_{11} = \frac{1}{16} \left(5-\sqrt{13 - 16x} \right) &\text{ on } x\in \left\{\frac{9}{16}\right\}\cup\left[\frac{3}{4},\frac{13}{16}\right],
\end{align}
\end{subequations}
note that the apparently exceptional point $x = \frac{9}{16}$, $\rho_{11} = \frac{3}{16}$ lies on the line $\rho_{11} = \frac{1}{32} \left(16 x-3 \right)$. To these we add the boundry values
\begin{subequations}
\label{eqn:rho11-con}
\begin{align}
  \label{eqn:rho-plus-con-1}\rho_{11} = 0 &\text{ with } \rho_{33}^+ \text{ and } x\in\left[0,\frac{1}{4}\right]\\
  \label{eqn:rho-plus-con-2}\rho_{11} = \frac{1}{2}\left(1-\sqrt{1-4x}\right) &\text{ with } \rho_{33}^+ \text{ and } x\in\left[0,\frac{1}{4}\right]\\
  \label{eqn:rho-plus-con-3}\rho_{11} = \frac{1}{2}\left(1+\sqrt{1-4x}\right) &\text{ with } \rho_{33}^+ \text{ and } x\in\left[0,\frac{1}{4}\right]\\
  \label{eqn:rho-plus-con-4}\rho_{11} = \frac{1}{2}\left(1+\sqrt{1-x}\right) &\text{ with } \rho_{33}^+ \text{ and } x\in\left[0,1\right]\\
  \label{eqn:rho-plus-con-5}\rho_{11} = \frac{1}{8}\left(4x-1\right) &\text{ with } \rho_{33}^+ \text{ and } x\in\left[\frac{1}{4},\frac{3}{4}\right]\\  
  \label{eqn:rho-plus-con-6}\rho_{11} = \frac{1}{2}\left(1-\sqrt{1-x}\right) &\text{ with } \rho_{33}^+ \text{ and } x\in\left[\frac{3}{4},1\right]\\
  \label{eqn:rho-minus-con-1}\rho_{11} = 0  &\text{ with } \rho_{33}^- \text{ and } x\in\left[0,\frac{1}{4}\right]\\
  \label{eqn:rho-minus-con-2}\rho_{11} = \frac{1}{8}(4x-1)  &\text{ with } \rho_{33}^- \text{ and } x\in\left[\frac{1}{4},\frac{3}{4}\right]\\
  \label{eqn:rho-minus-con-3}\rho_{11}= \frac{1}{2}(1-\sqrt{1-x}) &\text{ with } \rho_{33}^- \text{ and } x\in\left[0,\frac{3}{4}\right]\
\end{align}
\end{subequations}
the (locally) extremising values of $\rho_{11}$ are summarised in Figure \ref{fig:rho-plus-vals}. The values of $\var[\rho_{11}^+]{\opb}$ given these choices of $\rho_{11}$, and $\left(\Re{\rho_{13}}\right)^2 = \rho_{11}\rho_{33}^\pm$ are plotted in Figure \ref{fig:rho-minus-vals}.
\begin{figure}[ht]
  \begin{subfigure}[t]{0.4\textwidth}
    \includegraphics[height=0.65\linewidth]{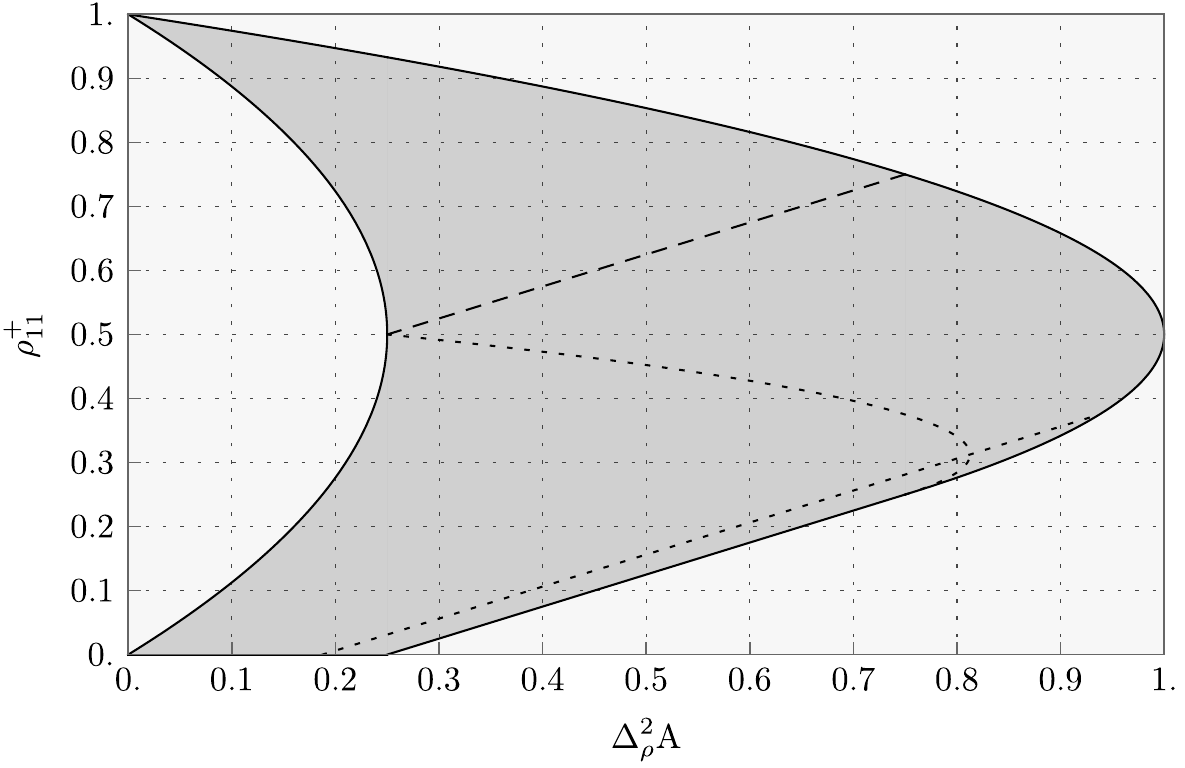}
    \caption{Range of $\rho_{11}$ given $\rho_{33} = \rho_{33}^+$. The dotted lines are the where $\Re{\rho_{13}}$ is maximal and the derivative of $\var[\rho^+]{\opb}$ with respect to $\rho_{11}$ is zero \eqref{eqn:rho-plus-zero-deriv-solns}, the dashed line is the where $\Re{\rho_{13}} = 0$ and the derivative of $\var[\rho^+]{\opb}$ with respect to $\rho_{11}$ is zero \eqref{eqn:rho-plus-max-zero-deriv-solns}.}\label{fig:rho-plus-vals}
  \end{subfigure}\quad
  \begin{subfigure}[t]{0.4\textwidth}
    \includegraphics[height=0.65\linewidth]{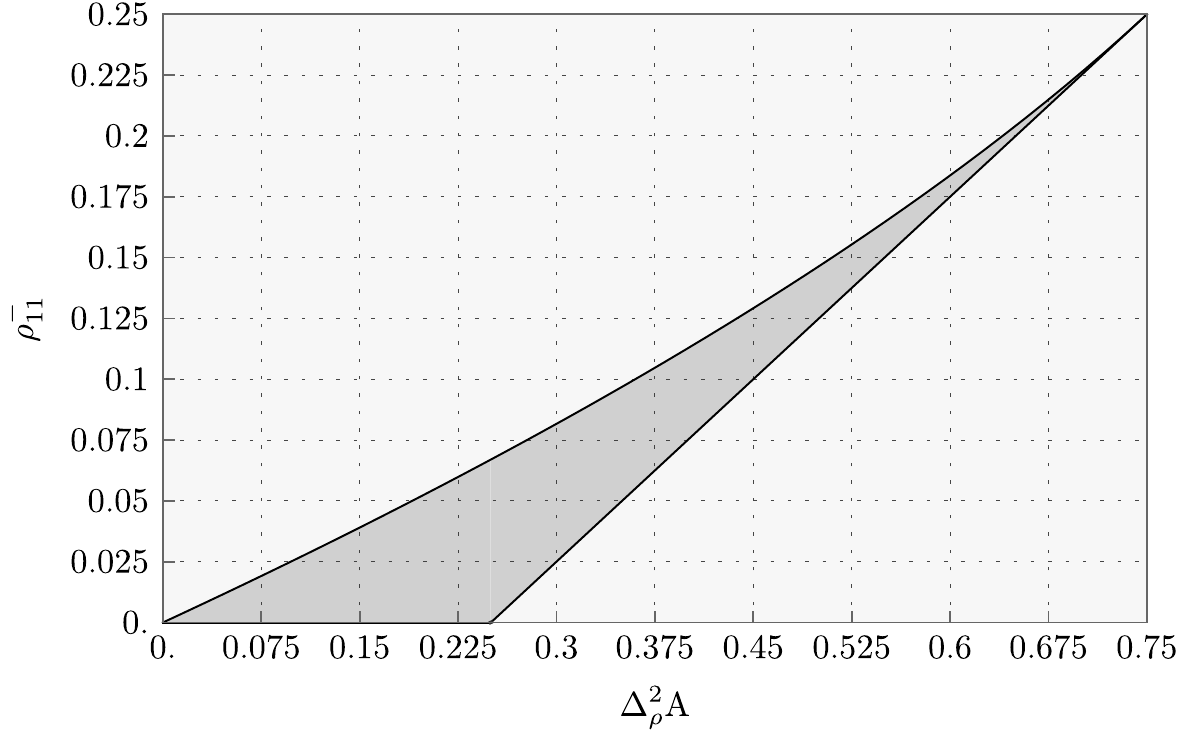}
    \caption{Range of $\rho_{11}$ given $\rho_{33} = \rho_{33}^-$. There are no local extrema other than the boundary curves.}\label{fig:rho-minus-vals}
      \end{subfigure}
      \caption{The filled region indicates the allowed values of $\rho_{11}$ as a function of $\var[\rho^+]{\opa}$ in each case. The solid lines are the boundary curves, given in  {\eqref{eqn:rho11-con}}}
\end{figure}

\subsection{Exploring maxima}
Here we consider the case $\left(\Re{\rho_{13}}\right)^2 = 0$. In this case
\begin{align}
  \var[\rho^\pm]{\opb} &= \rho_{11} + \rho_{33}^\pm\\
                       &= \frac{1}{2} \left(1-2 \rho_{11}\pm\sqrt{8 \rho_{11}-4 x+1}\right)\\
  \diff{\left(\var[\rho^\pm]{\opb}\right)}{\rho_{11}} &= -1 \pm \frac{2}{\sqrt{8\rho_{11} -4x+1}}\\
  \diff{\left(\var[\rho^\pm]{\opb}\right)}{\rho_{11}} &= 0 \iff \sqrt{8\rho_{11} -4x+1} = \pm 2.
                                                        \label{eqn:deriv-y-rhomax-is-zero}
\end{align}
There are no solutions for $\rho_{33}^-$, but $\rho_{33}^+$ has the solution
\begin{align}
  \label{eqn:rho-plus-max-zero-deriv-solns}
  \rho_{11} = \frac{1}{8}\left(3+4x\right),
\end{align}
which is always a valid solution for $\rho^+$ and never valid for $\rho^-$. To this we add the boundary values which are the same as those with $\left(\Re{\rho_{13}}\right)^2 = \rho_{11}\rho_{33}$, given in \eqref{eqn:rho11-con}.

\subsection{The bounding curves}
Comparing the local extrema we can now describe the full uncertainty region shown in figure \ref{fig:qutrit-uncertainty-region}
\begin{alignat}{2}
  & \var{\opb} = 1, \quad && \var{\opa}\in\left[0,\frac{1}{4}\right]\\
  & \var{\opb} = \frac{1}{8} \left(9-4\,\var{\opa}\right), \quad && \var{\opa}\in\left[\frac{1}{4},\frac{3}{4}\right]\\
  & \var{\opb} = \frac{1}{2} \left(1+\sqrt{1-\var{\opa}}\right), \quad && \var{\opa}\in\left[\frac{3}{4}, 1\right]\\
  & \var{\opb} = \frac{1}{2} \left(1-\sqrt{1-\var{\opa}}\right), \quad && \var{\opa}\in\left[\frac{15}{16}, 1\right]\\ 
  & \var{\opb} =  2 \left(\var{\opa}\right)^2-\frac{11}{4}\,\var{\opa}+\frac{153}{128} \quad && \var{\opa}\in\left[\frac{13}{16},\frac{15}{16}\right]\\
  & \var{\opb} = \frac{1}{8} \left(4\,\var{\opa}-1\right), \quad && \var{\opa}\in\left[\frac{1}{4}, \frac{13}{16}\right]\\
  & \var{\opb} = 1-\var{\opa}, \quad && \var{\opa}\in\left[\frac{15}{64},\frac{1}{4}\right]\\
  & \var{\opb} = \frac{1}{2}\left(1-\sqrt{1-\var{\opa}}\right) \quad && \var{\opa}\in\left[0,\frac{15}{64}\right]\\
  & \var{\opb} = \frac{1}{2} \left(1-\sqrt{1-4\,\var{\opa}}\right)\quad  && \var{\opa}\in\left[0,\frac{3}{16}\right]\\  
  & \var{\opb} = 1-4\,\var{\opa}, \quad && \var{\opa}\in\left[0,\frac{3}{16}\right].
  \label{eqn:gellmann-bounding-curves}
\end{alignat}

\end{document}